%% file: Manuscript.tex
\RequirePackage{fix-cm}
\documentclass[smallextended]{svjour3}       
\smartqed  
\usepackage{graphicx}
\usepackage{amssymb}
\usepackage{amsmath}
\makeatletter
\let\cl@chapter\undefined
\makeatother

\usepackage{cleveref}
\usepackage{nicefrac}
\usepackage[caption=false]{subfig} 
\usepackage{epstopdf}
\usepackage{color}
\usepackage{mathtools}
\usepackage{tikz}
\usepackage{pgfplots}
\pgfplotsset{compat=1.13}
\usetikzlibrary{plotmarks}
\usetikzlibrary{math,intersections}
\newcommand{\ie}{\textit{i.e.}\,}
\newcommand{\eg}{\textit{e.g.}\,}
\newcommand{\R}{\mathbb R}

\newcommand{\dd}{\;\mathrm{d} }
\newcommand{\ddx}{\;\mathrm{d} {\mathbf x}}
\newcommand{\ddsx}{\;\mathrm{d} s_{\mathbf x}}
\newcommand{\vect}[1]{{\mathbf #1}}

\newcommand{\pfrac}[2]{\frac{\partial #1 }{\partial #2} }

\newcommand{\projector}{{\mathbf P}}

\newcommand{\cuboid}{B}
\newcommand{\n}{\RealNormalvector}
\newcommand{\ctimes}{\cdot\!\!\times}
\newcommand{\pos}{{\mathbf x}}
\newcommand{\bodyforce}{\mathbf b}
\newcommand{\refArea}{U}
\newcommand{\parametrization}{\hat{\mathfrak g}}
\newcommand{\Surface}{\Omega}

\newcommand{\baseVector}{\hat{\mathbf g}}
\newcommand{\RealBaseVector}{{\mathbf g}}
\newcommand{\dualbaseVector}{{\hat {\mathbf g}}}

\newcommand{\indA}{\alpha}
\newcommand{\indB}{\beta}
\newcommand{\indC}{\gamma}
\newcommand{\indD}{\varphi}
\newcommand{\indE}{\varkappa}
\newcommand{\indF}{\xi}

\newcommand{\displacement}{\vect u}

\newcommand{\christbS}[2]{ \hat\Gamma_{#1}^{#2}}
\newcommand{\curvature}{\hat h}

\newcommand{\RealNormalvector}{\pmb{\nu}}
\newcommand{\RealWeingarten}{{{\mathbf H}}}
\newcommand{\ureal}{{\mathbf u}}

\newcommand{\RealChangeMetric}{{{\pmb \gamma}}}

\newcommand{\RealChangeCurvature}{{{\pmb \rho}}}
\newcommand{\RealTensor}{{\mathbf T}}
\newcommand{\realTensor}{{T}}
\newcommand{\paraMetric}{\hat{ G}}
\newcommand{\ParaMetric}{\hat{\mathbf G}}
\newcommand{\ParaNormalvector}{\hat{\pmb{\nu}}}
\newcommand{\ParaWeingarten}{{\hat{\mathbf H}}}
\newcommand{\upara}{{\hat{\mathbf u}}}
\newcommand{\paraChangeMetric}{{\hat{ \gamma}}}
\newcommand{\ParaChangeMetric}{{\hat{\pmb \gamma}}}
\newcommand{\paraChangeCurvature}{{\hat{ \rho}}}
\newcommand{\ParaChangeCurvature}{{\hat{\pmb \rho}}}
\newcommand{\ParaTensor}{\hat{\mathbf T}}
\newcommand{\um}[1]{#1}
\newcommand{\Nm}[1]{#1}
\newcommand{\Nmm}[1]{#1}

\newcommand{\appendexref}[1]{Appendix \ref{#1}}

\newcommand{\stuff}[1]{{\footnotesize \color{gray}}}

\graphicspath{{math_tex/}}

\journalname{Computational Mechanics}
\begin{document}

\title{A $C^1$-continuous Trace-Finite-Cell-Method for linear thin shell analysis on implicitly defined surfaces}


\titlerunning{Trace-Finite-Cell-Method for thin shell analysis}        

\author{Michael H. Gfrerer         
}


\institute{Michael H. Gfrerer \at
              Institute of Applied Mechanics, Graz University of Technology, Technikerstrasse 4, 8010 Graz, Austria \\
              Tel.: +43 (316) 873 - 7149 \\
              Fax: +43 (316) 873 - 7641 \\
              \email{gfrerer@tugraz.at}           
}

\date{Received: date / Accepted: date}

\maketitle

\begin{abstract}
A Trace-Finite-Cell-Method for the numerical analysis of thin shells is presented combining concepts of the TraceFEM and the Finite-Cell-Method. As an underlying shell model we use the Koiter model, which we re-derive in strong form based on first principles of continuum mechanics by recasting well-known relations formulated in local coordinates to a formulation independent of a parametrization. The field approximation is constructed by restricting shape functions defined on a structured background grid on the shell surface. As shape functions we use on a background grid the tensor product of cubic splines. This yields $C^1$-continuous approximation spaces, which are required by the governing equations of fourth order. The parametrization-free formulation allows a natural implementation of the proposed method and manufactured solutions on arbitrary geometries for code verification. Thus, the implementation  is verified by a convergence analysis where the error is computed with an exact manufactured solution. Furthermore, benchmark tests are investigated.	
	
\end{abstract}

\section{Introduction}
\input{s_intro}
\section{Notation and geometric preliminaries}
\input{s_geometry}
\section{The linear thin shell problem}
\input{s_shellModel}
\section{$C^1$-Trace-Finite-Cell-Method}
\input{s_c1TraceFEM}
\section{Numerical results}
\input{s_numerical_results}
\section{Conclusions}
We have developed a $C^1$-continuous finite element method for thin shells with mid-surface given as the zero level-set of a scalar function. In order to achieve the continuity of the discretization, concepts of the TraceFEM and the Finite-Cell-Method are combined. In particular the shape functions on the shell surface are obtained by restriction of tensor-product cubic Hermite splines on a structured background mesh. In order to allow a natural implementation, the underlying shell model is formulated in a parametrization-free way. Furthermore, the strong form of the governing equations are given. This allows to obtain manufactured solutions on arbitrary geometries. Thus, the implementation of the proposed method is verified by a convergence analysis where the error is computed with an exact manufactured solution.

In the present method, the shape functions on the shell surface are linearly dependent. In order to avoid a singular system matrix, a stabilization term can be used. In the presented method such a stabilization is avoided. However, it is necessary to use the direct solver suitable for under-determined linear equation systems from the SuiteSparse project. We investigated three strategies to include the boundary conditions. These are the penalty method, the Lagrange multiplier method, and the null-space method. In the numerical experiments we have observed that the penalty method and the null-space method give reliable results. However, the Lagrange multiplier method suffers from instabilities in some examples, which should be further investigated. In future work, it would be also interesting to use iterative solvers in contrast to the used direct solver.

In contrast to thin shells, for Reissner-Mindlin shells only $C^0$-continuous shape functions are commonly used. In order to avoid transverse shear locking, in \cite{long2012shear,echter2013hierarchic} an hierarchic concept of shell models is presented. This approach has the advantage that transverse shear locking is eliminated on the continuous formulation level, independent of a particular discretization, but requires $C^1$-continuous shape functions. An extension of the present work to Mindlin-Reissner shells with implicitly defined mid-surface seems possible and would be worth to investigate. 

\begin{acknowledgements}
The author thanks Thomas-Peter Fries from the Institute of Structural Analysis at TU Graz, and Helmut Gfrerer from the Institute of Computational Mathematics at JKU Linz for valuable discussions on the topic of the paper.
\end{acknowledgements}

%
\section*{Conflict of interest}
The author declares that he has no conflict of interest.
\input{s_appendix}
\bibliographystyle{spmpsci}      
\bibliography{literature}

\end{document}

%% file: s_intro.tex
Due to the superior load-carrying capabilities, the mechanical analysis of shells is of great interest in engineering. A large literature body exists on the formulation of shell models. We refer to \cite{naghdi1981finite,pietraszkiewicz1989geometrically,bischoff2004} and references therein for an overview. In the present paper, we consider the classical Koiter model \cite{koiter1966nonlinear,ciarlet2006}, which is one of the most popular models for thin shells. The Koiter shell model has been justified by asymptotic analysis in \cite{Ciarlet1996c}, in the sense of being a reasonable approximation to the full 3D problem of a thin shell-like body. Existence and uniqueness results can be found in \cite{blouza1999existence,ciarlet2006}.
 
Classically, for theoretical treatment it is assumed that the shell mid-surface is represented by a global parametrization. However, for the numerical treatment typically the mid-surface is approximated by, possibly curved, finite elements, \ie represented by a collection of local parametrizations. In contrast to these representations, we consider the case where the mid-surface is represented implicitly as the zero-level set of a scalar function $\phi(\pos)$, see \Cref{fig::intro_ex} for the illustration of some examples.  We refer to the review article \cite{dziuk2013finite} for an overview of finite element methods for problems on such surfaces. In the classical surface finite element method the discretization of the unknown field relies on the higher order or exact meshing (local parametrization) of the surface \cite{demlow2009higher,gfrerer2018b}. In contrast to this, in the proposed method the discretization of the displacement field does not rely on parametrizations. 
Therefore, we provide a throughout derivation of the governing equations based on first principles of continuum mechanics independent of a parametrization. This allows a natural implementation of the method and also the construction of manufactured solutions for code verification on arbitrary geometries. Equivalent derivations relying on a parametrization can be found in \eg \cite{basar1985} and \cite{sauer2017theoretical}. We remark that membrane and thin shell formulations without relying on a parametrization with a mathematical focus can be found in \cite{gurtin1975continuum,delfour1997differential}. For a treatment from an engineering perspective we refer to \cite{van2015finite,schollhammer2019kirchhoff}.   

\begin{figure}[th]
\subfloat[sphere \label{fig::intro_ex_A}]
{\includegraphics[width=0.22\textwidth]{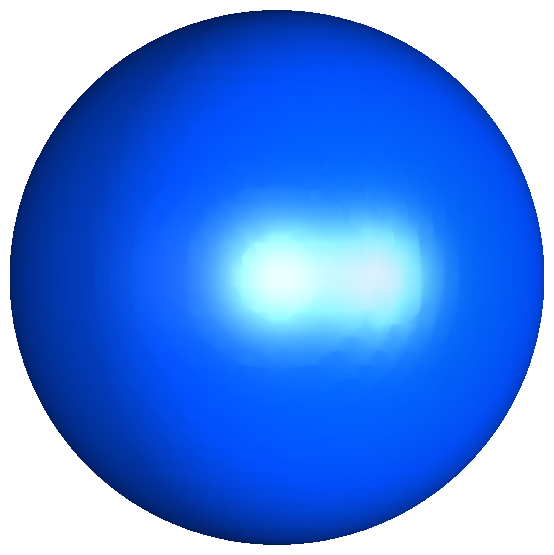}}
\hfill
\subfloat[torus	\label{fig::intro_ex_B}]
{\includegraphics[width=0.22\textwidth]{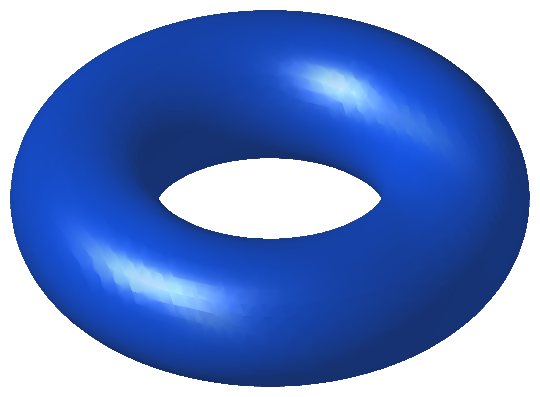}}
\hfill
\subfloat[cylinder \label{fig::intro_ex_C}]
{\includegraphics[width=0.22\textwidth]{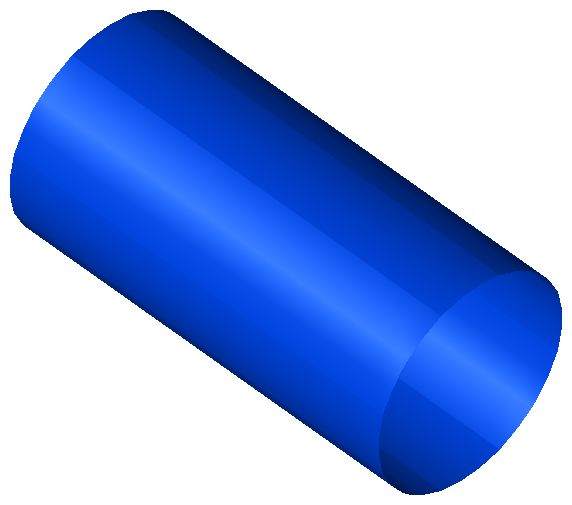}}
\hfill
\subfloat[gyroid \label{fig::intro_ex_D}]
{\includegraphics[width=0.22\textwidth]{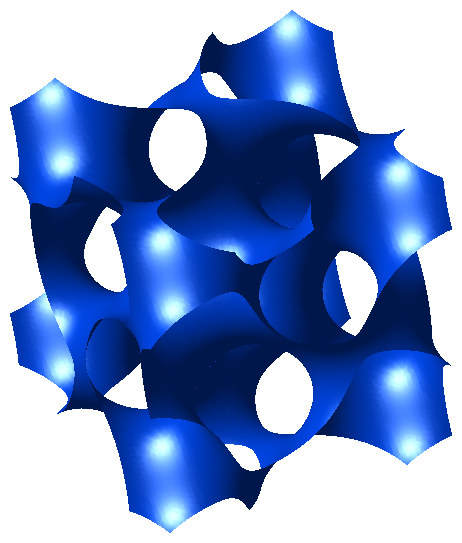}}
\caption{Examples of implicitly defined surfaces. The surfaces are defined by the level-set functions (a) $\phi=x^2+y^2+z^2-r^2$, (b) $\phi = (x^2+y^2+z^2+R^2-r^2)^2- 4R^2(x^2+y^2)$, (c) $\phi = x^2+y^2 - r^2$ (d) $\phi = \sin x\cos y+\sin y\cos z+\sin z\cos x$   }
\label{fig::intro_ex}
\end{figure}

One of the main difficulties in developing finite element methods for thin shells is the construction of $C^1$-continuous approximation spaces. For general unstructured meshes it is not possible to ensure $C^1$-continuity with only local polynomial shape functions and the nodal degrees of freedom consist of displacements and slopes only \cite{Zienkiewicz_Taylor_2000}. However, different non-standard triangular for developed thin plate bending are the Argyris element \cite{argyris1968tuba,Dominguez_Sayas_2008}, the Bell element \cite{bell1969refined} or the Clough-Tocher macrotriangle \cite{clough1965finite}. A further possibility to construct $C^1$-continuous approximation spaces on general space triangulation relies on sophisticated techniques from subdivision surfaces \cite{cirak2000subdivision}. However, on a structured quadrilateral mesh the Bogner-Fox-Schmit element \cite{bogner1965generation} is a simple conforming element. The constraint of a structured quadrilateral mesh can be partially overcome by introducing a smooth mapping of the geometry \cite{gfrerer2018code}. This idea can be realized in an isoparametric way by the use of splines for the geometry mapping and for the discretization of the displacement field \cite{Kiendl_Bletzinger_Linhard_Wüchner_2009}. The general difficulty of constructing $C^1$-continuous approximation spaces led to approaches where the $C^1$-continuity requirement is circumvented. Among them we mention discrete Kirchhoff elements \cite{Batoz_Zheng_Hammadi_2001,Areias_Song_Belytschko_2005} where the Kirchhoff constraint is enforced only at discrete points, the use of shear-deformable (Reissner-Mindlin) shell theory, were only $C^0$-continuity approximation spaces are required, mixed methods \cite{Rafetseder_Zulehner_2019,neunteufel_hellanherrmannjohnson_2019}, continuous/discontinuous Galerkin methods \cite{engel2002continuous,Hansbo_Larson_2017} and others. 

In the present paper, we combine ideas from unfitted finite element methods and the Bogner-Fox-Schmit element. Following the idea of the TraceFEM \cite{Olshanskii_Reusken_Grande_2009,Olshanskii_Reusken_2017} (see also CutFEM \cite{burman2015cutfem,burman2018cut}) the approximation of the displacement field is constructed by restricting shape functions defined on a background mesh on the shell surface. In particular, we follow the idea of the Finite-Cell-Method \cite{Parvizian_Düster_Rank_2007,schillinger2015finite} and use a structured grid where the simple tensor product of three uni-variant cubic spline shape functions leads to a $C^1$-continuous approximation in 3D (like the Bogner-Fox-Schmit element in 2D). Therefore, the shape functions for approximation of the displacement field on the shell mid-surface are $C^1$-continuous. We remark that cut Bogner-Fox-Schmit elements for thin plates were proposed and analyzed in \cite{burman_cut_2019}. Therefore, the proposed method can be seen as an extension of the work \cite{burman_cut_2019} from plates to curved shells.

One challenge in unfitted finite element methods is the efficient integration on the problem domain \cite{olshanskii2016numerical}. During the preparation of the present paper it turned out that due to a gowning number of constraints for finer meshes the strategy developed in \cite{gfrerer2018a} is not applicable in the present situation. Therefore, we have implemented the quadrature strategy developed in \cite{saye2015high}.

The implementation of the proposed method is verified by a convergence analysis where the error is computed with an exact manufactured solution. Furthermore, the capabilities of the method are shown in two standard and one non-standard benchmark tests.

%% file: s_geometry.tex
The underlying assumption in shell analysis is that the computational domain has a small extension with respect to one coordinate. Thus, we assume that it is located around a two-dimensional mid-surface $\Surface$ which is embedded in $\R^3$. In the present paper we assume that the mid-surface is defined implicitly as the zero-level set of a function $\phi: \R^3  \rightarrow \R$ inside a cuboid $\cuboid \subset \R^3$,
\begin{equation}\label{eq:implicitRepresentation}
\Surface = \{\pos\in\cuboid|\,\phi(\pos) = 0 \}.
\end{equation}
The boundary of $\Surface$ is denoted $\Gamma$, the surface normal vector is denoted by $\RealNormalvector$, and the normal vector tangential to the surface on a boundary point is denoted by $\pmb \mu$, see \Cref{fig::problem}.
\begin{figure}[ht]
	\centering
	\includegraphics[width=0.3\textwidth]{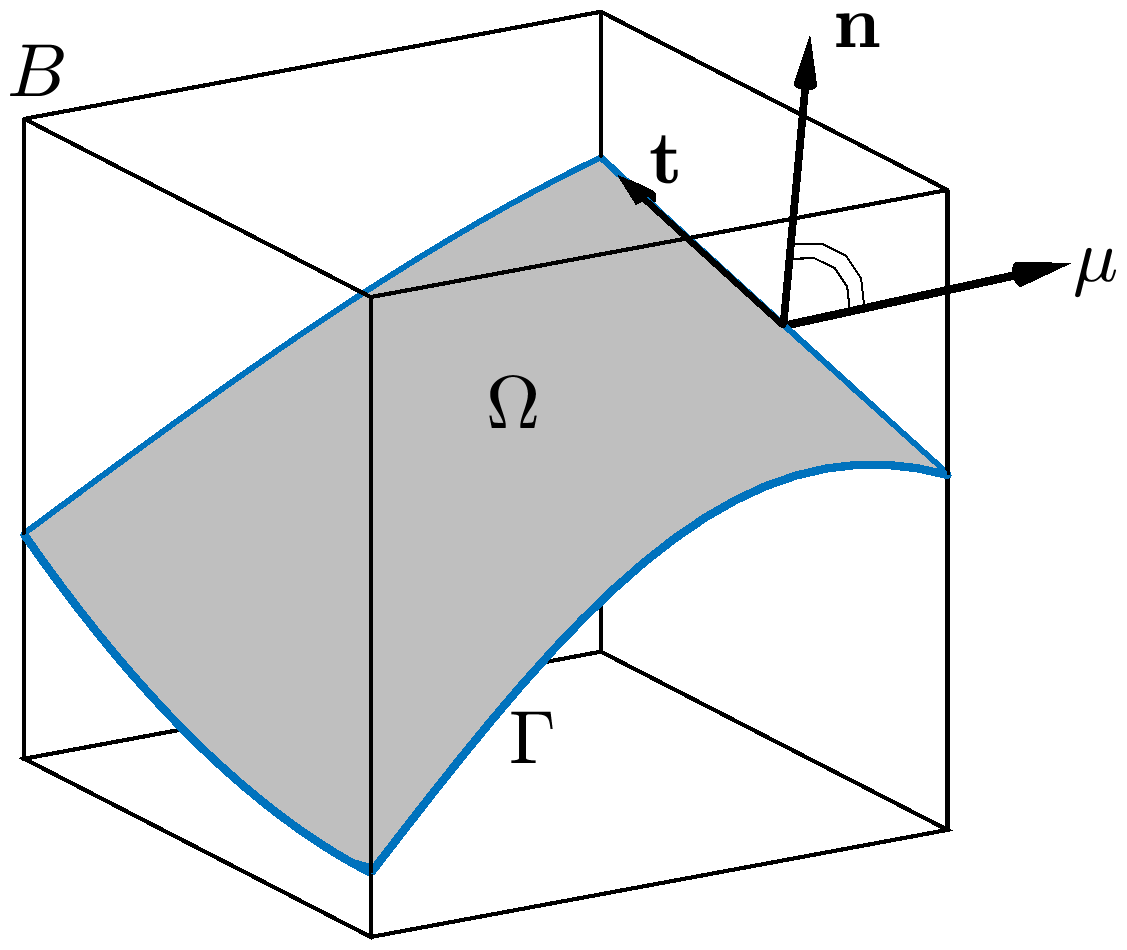}
	\caption{Illustration and notation of the geometric setting}
	\label{fig::problem}
\end{figure} 
We assume that $\Surface$ is regular such that
\begin{equation}
\nabla \phi(\pos) \neq \mathbf 0,
\end{equation}
holds in the neighborhood of $\Surface$, where $\nabla$ denotes the usual gradient of some scalar-valued function $f:\R^3\rightarrow\R$,
\begin{equation}\label{eq::gradient}
\nabla f(\pos) =  f_{,i} \,\mathbf e^i :=  \frac{\partial f(\mathbf x)}{\partial x_i} \mathbf e^i
\end{equation}
with the Cartesian coordinates $\pos = (x_1,x_2,x_3)$ and the standard Cartesian orthonormal basis $\{\mathbf e^1,\,\mathbf e^2,\,\mathbf e^3\}$. Here, and in the following, the Einstein summation convention applies. Whenever an index occurs once in an upper position and in a lower position we sum over this index, where Latin indices $i,j,\dots$ take the values $1,2,3$ whereas Greek indices $\indA, \indB,\dots$ take the values $1,2$. Let $\mathcal T$ be some tensor space of the form $\R^3\otimes\dots \otimes \R^3$.
 In the following we also use the generalization of the gradient for scalar-valued functions \eqref{eq::gradient} to tensor-valued functions $\RealTensor: \R^3  \rightarrow \mathcal T$,
\begin{equation}
\nabla \RealTensor(\pos) =  \RealTensor_{,i} \otimes\mathbf e^i.
\end{equation}
\subsection{Differential geometry of implicitly defined surfaces}
Given a implicit representation \eqref{eq:implicitRepresentation} of the surface $\Surface$, we can compute the unit normal vector to the surface by 
\begin{equation}\label{eq::real_normal_vector}
\RealNormalvector(\pos) = \frac{\nabla \phi(\pos)}{||\nabla \phi(\pos)||},
\end{equation}
and are able to define the tangential projector,
\begin{equation}
\projector = \mathbf I - \RealNormalvector \otimes \RealNormalvector.
\end{equation}
Furthermore, the extended Weingarten map is given by
\begin{equation}\label{eq::real_weingarten}
\RealWeingarten  = - \nabla \RealNormalvector \cdot \projector = -\frac{\mathbf P \cdot \nabla\nabla\phi \cdot \mathbf P}{||\nabla \phi||},
\end{equation}
and the mean curvature $H$ is defined as 
\begin{equation}
H = \text{tr}(\RealWeingarten) = \RealWeingarten : \projector.
\end{equation}
\subsection{Differential geometry of parametrized surfaces}
We briefly review the differential geometry of parametrized surfaces. For details we refer to \eg \cite{ciarlet2006}. Although the mid-surface $\Surface$ is assumed to be given implicitly, at least a local parametric representation is guaranteed to exist by the implicit function theorem. This justifies to consider parametrizations $\parametrization:\;\refArea \subset \R^2 \rightarrow \Surface$ with the parameter domain $\refArea$ for theoretical considerations. Given the parametrization $\parametrization(\theta^1,\theta^2)$, we can define the two covariant base vectors $\baseVector_\indA := \pfrac{\parametrization}{\theta^\indA}$, which span the tangent plane to $\Surface$. With the base vectors we can define the unit normal vector 
\begin{equation}
\ParaNormalvector(\theta^1,\theta^2)  = \frac{\baseVector_1(\theta^1,\theta^2) \times \baseVector_2(\theta^1,\theta^2)}{||\baseVector_1(\theta^1,\theta^2) \times \baseVector_2(\theta^1,\theta^2)||},
\end{equation}
and the covariant coefficients of the metric $\paraMetric_{\indA\indB} = \baseVector_\indA \, \cdot \, \baseVector_\indB$. The contravariant coefficients of the metric are given by $[\paraMetric^{\indA\indB}]=[\paraMetric_{\indA\indB}]^{-1}$, where $[\paraMetric_{\indA\indB}]$ is the coefficient matrix. The contravariant base vectors can then be computed by $\dualbaseVector^\indA = \paraMetric^{\indA\indB} \baseVector_\indB$. The covariant coefficients of the Weingarten map $\ParaWeingarten = \curvature_{\alpha\beta} \, \baseVector^\indA \otimes \baseVector^\indB$ are given by
\begin{equation}
\curvature_{\alpha\beta} = - \baseVector_{\alpha} \cdot \ParaNormalvector_{,\beta},
\end{equation}
and obey the symmetry relation $\curvature_{\alpha\beta} = \curvature_{\beta\alpha}$. The mean curvature $\hat H$ is given by
\begin{equation}
\hat H = \curvature_{\alpha}^{\alpha} = \curvature_{\alpha\beta} \paraMetric^{\indB\indA}.
\end{equation}
Furthermore, the derivatives of the base vectors are given by
\begin{equation}\label{eq::derivativeBaseVector}
\begin{aligned}
\baseVector_{\indA,\indB} = \Gamma_{\indA\indB}^\indC \baseVector_{\indC} + \curvature_{\indA\indB} \ParaNormalvector, \\
\baseVector_{,\indB}^{\indA} = -\Gamma_{\indB\indC}^\indA \baseVector^{\indC} + \curvature^{\indA}_{\indB} \ParaNormalvector,
\end{aligned}
\end{equation}
with the surface Christoffel symbols of the second kind defined by
\begin{equation}
\christbS{\indA\indB}{\indC} = \baseVector^\indC \cdot \baseVector_{\indA,\indB}. 
\end{equation}
\textbf{Remark:} In our notion a hat over a quantity refers to a dependency on the parametric coordinates $(\theta^1,\theta^2)\in \refArea$, whereas no hat refers to a dependency on $\pos \in \R^3$. 
\subsection{Relations between parameter space and embedding space}
The field $\upara(\theta^1,\theta^2)$ defined on the parameter space is related to the field $\ureal(\pos)$ defined on the embedding space $\R^3$ by
\begin{equation}\label{eq::relation}
\begin{aligned}
\upara(\theta^1,\theta^2) &= \ureal(\pos)\circ \parametrization(\theta^1,\theta^2) = \ureal(\parametrization(\theta^1,\theta^2)) 
\end{aligned}
\end{equation}
By applying the chain rule we find that the first and second derivatives are related by
\begin{equation}\label{eq::relationDerivative}
\begin{aligned}
\upara_{,\theta} &= (\nabla \ureal \circ \parametrization) \cdot \baseVector_\theta, 
\end{aligned}
\end{equation}
and
\begin{equation}\label{eq::relationDerivativeSecond}
\begin{aligned}
\upara_{,\theta\tau} &= (\nabla\nabla \ureal \circ \parametrization)\cdot \baseVector_\tau \cdot \baseVector_\theta  + (\nabla \ureal \circ \parametrization) \cdot \baseVector_{\theta,\tau}\\
&= (\nabla\nabla \ureal \circ \parametrization)\cdot \baseVector_\tau \cdot \baseVector_\theta  + (\nabla \ureal \circ \parametrization) \cdot (\Gamma_{\theta\tau}^\indA \, \baseVector_\indA + h_{\theta\tau} \,\RealNormalvector) \\
&= (\nabla\nabla \ureal \circ \parametrization): (\baseVector_\tau \otimes \baseVector_\theta)  + (\nabla \ureal \circ \parametrization) \cdot (\Gamma_{\theta\tau}^\indA \, \baseVector_\indA + h_{\theta\tau} \,\RealNormalvector).
\end{aligned}
\end{equation}
Furthermore, we have the following relations summarized in the following lemma.
\begin{lemma}\label{thm::projectorMetric}
	The metric tensor $\ParaMetric = \paraMetric_{\indA\indB} \,\baseVector^\indA \otimes \baseVector^\indB = \baseVector_\indA \otimes \baseVector^\indA$ and the projector $\mathbf P$ are related by
	\begin{equation}
	\ParaMetric = \mathbf P \circ \parametrization.
	\end{equation}
	For the Weingarten map we have the relation 
	\begin{equation}
	\ParaWeingarten = {\mathbf H}\circ\parametrization.
	\end{equation}
	The proof can be found in \appendexref{sec::proofs}.
\end{lemma}
\subsection{Surface gradient}
The surface gradient of a tensor-valued function represented with respect to parametric coordinates by the map $\hat {\mathbf f}\colon \refArea\rightarrow \mathcal T$ is given by
\begin{equation}
\nabla_\Surface \hat {\mathbf f} = \hat {\mathbf f}_{,\indA} \otimes \baseVector^\indA.
\end{equation}
\begin{lemma}
For the representation ${\mathbf f}\colon \R^3\rightarrow \mathcal T$ the surface gradient is given by
\begin{equation}
\nabla_\Surface  {\mathbf f} =  \nabla \mathbf f \cdot \projector.
\end{equation}	
\end{lemma}
\begin{proof} Using the relation between the projector and the metric tensor and \eqref{eq::relationDerivative} we have
\begin{align}
\nabla_\Surface  {\mathbf f} \circ \parametrization = (\nabla \mathbf f \circ \parametrization ) \cdot (\baseVector_\indA \otimes \baseVector^\indA) = \hat {\mathbf f}_{,\indA} \otimes \baseVector^\indA.
\end{align}	
\end{proof}
\subsection{Surface divergence}
We define the surface divergence as the adjoint operator to the surface gradient \cite{rosenberg1997laplacian}. Therefore, on an Riemannian manifold we have in local coordinates
\begin{equation}
\text{div} \ParaTensor  = \frac{1}{\sqrt{\det\ParaMetric}} \left(\ParaTensor\cdot\baseVector^\indA\sqrt{\det\ParaMetric}\right)_{,\indA},
\end{equation}
where we use the notation $\det\ParaMetric = \det([\paraMetric_{\indA\indB}])$. The next lemma gives the simpler representation for the surface divergence in case of a surface embedded in $\R^3$. 
\begin{lemma}\label{thm::surfaceDivergence}
	On a surface $\Omega\subset \R^3$ parametrized by $\parametrization:\refArea\rightarrow \Omega$ the surface divergence of a tensor-valued function represented by $\ParaTensor\colon \refArea\rightarrow \mathcal T$ is given by
	\begin{equation}
	\textup{div}\ParaTensor  = \ParaTensor_{,\indA}\cdot\baseVector^\indA + H \, \ParaTensor \cdot \ParaNormalvector.
	\end{equation}
	Furthermore, for the representation $\RealTensor\colon \R^3\rightarrow \mathcal T$ we have 
	\begin{equation}
	\textup{div} \mathbf T = \nabla \mathbf T:\projector + H\,\mathbf T \cdot \RealNormalvector.
	\end{equation}
	The proof can be found in \appendexref{sec::proofs}.
\end{lemma}
In the following lemma we collect product rules for the divergence operator.
\begin{lemma}\label{thm::productRules}
	Let $\mathbf v\times\mathbf T$ be the cross product of a vector $\mathbf v = v_i \mathbf e^i$ and a second order tensor $\mathbf T = T_{lk} \mathbf e^l \otimes \mathbf e^k$ defined by
	\begin{equation}
	\mathbf v\times\mathbf T = v_iT_{lk} (\mathbf e^i\times\mathbf e^l) \otimes \mathbf e^k,
	\end{equation}
	and $\mathbf V \ctimes\mathbf T$ the scalar-cross product of two second order tensors $\mathbf V = V_{ij} \mathbf e^i \otimes \mathbf e^j$ and $\mathbf T = T_{lk} \mathbf e^l \otimes \mathbf e^k$ defined by
	\begin{equation}
	\mathbf V\ctimes\mathbf T = V_{ij}T_{lk} (\mathbf e^j\cdot\mathbf e^l)(\mathbf e^i\times\mathbf e^k).
	\end{equation}
	Then, the following product rules hold
	\begin{align}
	\textup{div}(\mathbf v\times\mathbf T) &= \mathbf v\times\textup{div}(\mathbf T) + \nabla_\Omega \mathbf v\ctimes\mathbf T^\top, \label{eq::divergenceProductCross}\\
	\textup{div}(\mathbf v\cdot\mathbf T) &= \mathbf v\cdot\textup{div}(\mathbf T) + \nabla_\Omega \mathbf v:\mathbf T^\top . \label{eq::divergenceProductDot}
	\end{align}
	The proof can be found in \appendexref{sec::proofs}. 
\end{lemma}

\stuff{
For a second order tensor we have
\begin{equation}
\begin{aligned}
\text{div} \mathbf T &= T^{\indA\indB}_{,\indB} \baseVector_{\alpha} + T^{\indA\indB} \baseVector_{\indA,\indB}  + T^{\indA\indB} \baseVector_{\indA} (\baseVector_{\indB,\indC}\cdot \baseVector^{\indC}) \\
&=T^{\indA\indB}_{,\indB} \baseVector_{\indA} + T^{\indA\indB} \Gamma_{\indA\indB}^{\indC} \baseVector_{\indC} + T^{\indA\indB} h_{\indA\indC}\RealNormalvector  + T^{\indA\indB} \Gamma_{\indB\indC}^{\indC}\baseVector_{\indA} 
\end{aligned}
\end{equation}
Also
\begin{equation}
\begin{aligned}
\text{div} (\mathbf a\otimes \mathbf b) &= \mathbf a_{,\indA} (\mathbf b \cdot \baseVector^{\alpha}) + \mathbf a( \mathbf b_{,\indA}\cdot \baseVector^{\alpha}) + h_\indA^\indA \mathbf a(\mathbf b \cdot \RealNormalvector)
\end{aligned}
\end{equation}
and
\begin{equation}
\begin{aligned}
\text{div} (\RealNormalvector\otimes \mathbf S) &= 
\RealNormalvector_{,\indA} (\mathbf S \cdot \baseVector^{\alpha}) + \RealNormalvector (\mathbf S_{,\indA}\cdot \baseVector^{\alpha}) \\
&= -h_\indA^\indB  S^\indA \baseVector_{\indB} + \RealNormalvector (S^\indB_{,\indA} \baseVector_{\indB} + S^\indB\baseVector_{\indB,\indA} ) \cdot \baseVector^{\alpha} \\
&= -h_\indA^\indB  S^\indA \baseVector_{\indB} + \RealNormalvector (S^\indA_{,\indA} + S^\indB \Gamma_{\indA\indB}^\indA ) 
\end{aligned}
\end{equation}
\begin{equation}
\begin{aligned}
\text{div} (\RealNormalvector\otimes \text{div} \mathbf T) &= 
\RealNormalvector_{,\indA} (\text{div} \mathbf T \cdot \baseVector^{\alpha}) +\RealNormalvector ( (\text{div} \mathbf T)_{,\indA}\cdot \baseVector^{\alpha}) + h_\indA^\indA \RealNormalvector( \text{div} \mathbf T \cdot \RealNormalvector) \\
&= -h_\indA^\indB \baseVector_{\indB}\left(T^{\indA\indC}_{,\indC}+T^{\indC\indB} \Gamma_{\indC\indB}^{\indA}+T^{\indA\indC} \Gamma_{\indC\indD}^{\indD}\right) \\
&+ \RealNormalvector \left[T^{\indA\indB}_{,\indB\indA} + T^{\indC\indB}_{,\indB} \Gamma_{\indC\indA}^\indA + T^{\indA\indB}_{,\indA}(\Gamma_{\indA\indB}^\indA+\Gamma_{\indB\indC}^{\indC}) + T^{\indA\indB}(\Gamma_{\indA\indB}^{\indC} \baseVector_{\indC} + h_{\indA\indC}\RealNormalvector  + \Gamma_{\indB\indC}^{\indC}\baseVector_{\indA})_{,\indA}\cdot \baseVector^{\alpha}  \right] \\
&+ h_\indA^\indA \RealNormalvector T^{\indA\indB} h_{\indA\indC}
\end{aligned}
\end{equation}
}
\subsection{Integral identities}
For further use we introduce the surface divergence theorem for a tensor-valued function $\RealTensor$
\begin{equation}\label{eq::divergenceTheorem}
 \begin{aligned}
 \int_\Omega \text{div} \,\RealTensor \ddx = \int_{\Gamma} \RealTensor \cdot \pmb \mu \ddsx.  
 \end{aligned}
\end{equation} 
Using \eqref{eq::divergenceProductDot} and \eqref{eq::divergenceTheorem}, the integration by parts formula for a vector $\mathbf v$ and a second order tensor $\RealTensor$ reads
\begin{equation}
\begin{aligned}
\int_\Surface\mathbf v\cdot \text{div}\mathbf T \ddx = \int_{\Gamma} \mathbf v \cdot\mathbf T \cdot \pmb \mu \dd s_x - \int_\Surface \nabla_\Surface \mathbf v : \mathbf T^\top  \ddx. 
\end{aligned}
\end{equation}

%% file: s_shellModel.tex
In this section we derive the governing equations of linear thin shells from first principles of continuum mechanics. Furthermore, we show the equivalence to the linear Koiter model formulated as a minimization problem.
\subsection{Shell kinematics}
The kinematics of the surface $\Surface$ is described by the change in metric tensor and the change in curvature tensor. In the present paper we focus on the linear theory and use the linearized change in metric tensor $\ParaChangeMetric = \paraChangeMetric_{\alpha\beta} \,\baseVector^\indA \otimes \baseVector^\indB$ and the linearized change in curvature tensor $\ParaChangeCurvature = \paraChangeCurvature_{\alpha\beta} \,\baseVector^\indA \otimes \baseVector^\indB$. The respective covariant components are given by \cite{ciarlet2006,blouza1999existence}
\begin{equation}
\paraChangeMetric_{\alpha\beta}(\upara) = \frac{1}{2} (\upara_{,\beta} \cdot \baseVector_\alpha + \upara_{,\alpha} \cdot \baseVector_\beta),
\end{equation}
and
\begin{equation}
\paraChangeCurvature_{\alpha\beta}(\upara) = \ParaNormalvector\cdot\left(\upara_{,\alpha\beta} - \Gamma_{\alpha\beta}^\sigma \upara_{,\sigma}\right).
\end{equation}
The next lemma establishes the representations for $\RealChangeMetric = \ParaChangeMetric \circ \parametrization^{-1}$ and $\RealChangeCurvature = \ParaChangeCurvature \circ \parametrization^{-1}$.  
\begin{lemma}\label{thm::changeTensors}
	For the linearized change in metric tensor we have the representation 
	\begin{equation}
	\RealChangeMetric(\ureal) = \frac{1}{2} \projector \cdot(\nabla \ureal + (\nabla \ureal)^\top) \cdot\projector, 
	\end{equation}
	and for the linearized change in curvature tensor we have the representation 
	\begin{equation}
	\pmb \rho(\ureal) =\projector\cdot(\RealNormalvector\cdot\nabla\nabla \mathbf u)\cdot\projector  - ( \RealNormalvector\cdot\nabla \ureal \cdot \RealNormalvector) \RealWeingarten.
	\end{equation}
	The proof can be found in \appendexref{sec::proofs}.
\end{lemma}
\subsection{Stress and moment tensors}
We define the traction vector $\mathbf t(\pmb \mu)$ on a cut defined by the boundary normal $\pmb \mu$ tangential to the surface. Due to Cauchy’s theorem we have the representation 
\begin{equation}
\mathbf t(\pmb \mu) = \pmb \sigma \cdot \pmb \mu, 
\end{equation}
with the stress tensor $\pmb \sigma$.
We decompose the stress tensor $\pmb \sigma$ in a tangential and a normal part,
\begin{equation}\label{eq::stressTensor}
\pmb \sigma = \mathbf N + \RealNormalvector\otimes\mathbf S, 
\end{equation}
with the tangential stress tensor $\mathbf N = N^{\indA\indB} \RealBaseVector_\indA\otimes\RealBaseVector_\indB$ and the vector $\mathbf S = S^{\indA} \RealBaseVector_\indA$ related to transverse shear. Analogously, we have a moment vector $\mathbf m(\pmb \mu)$ on the cut defined by $\pmb \mu$, which can be expressed as 
\begin{equation}
 \mathbf m(\pmb \mu) = \n\times (\mathbf M \cdot \pmb \mu), 
\end{equation} 
 with the tangential moment tensor $\mathbf M$. 
\subsection{Equilibrium of forces}
The equilibrium of forces states that the sum of the resulting force of boundary traction and the resultant force from the surface loading vanishes,  
\begin{equation}
\begin{aligned}
\int_{\Gamma} \mathbf t \ddsx + \int_\Omega \mathbf b \ddx = 0.
\end{aligned}
\end{equation}
Applying the surface divergence theorem \eqref{eq::divergenceTheorem} results in 
\begin{equation}
\begin{aligned}
\int_{\Omega} \text{div}\,\pmb\sigma + \bodyforce \ddx = 0.
\end{aligned}
\end{equation}
Due to the fact that $\Omega$ is arbitrary the local force equilibrium reads
\begin{equation}\label{eq::local_force_equilibrium}
 \text{div}\,\pmb\sigma + \mathbf b  = 0.
\end{equation}
\subsection{Equilibrium of moments}
The equilibrium of moments states that the sum of boundary moments, the moments of boundary tractions, and the moments due to surface loads vanishes,
\begin{equation}
\begin{aligned}
\int_{\Gamma}\mathbf m + \mathbf x \times \mathbf t \dd s_x + \int_{\Surface}\mathbf x \times \mathbf b \ddx = 0.
\end{aligned}
\end{equation}
The following lemma summarizes the consequences of the equilibrium of moments.
\begin{lemma}\label{thm::equilibirumMoments}
	For $\RealTensor = \realTensor_{ij} \mathbf e^i \otimes \mathbf e^j$ let $[\RealTensor]_{\times} = \realTensor_{ij}\mathbf e^i \times \mathbf e^j$. The equilibrium of moments is fulfilled if 
	\begin{equation}\label{eq::symetryMoment}
	[-\mathbf H\cdot\mathbf M + \mathbf N^\top]_{\times}  = 0,
	\end{equation}
	and
	\begin{equation}\label{eq::transverseForce}
	\mathbf S = \projector\cdot\textup{div}(\mathbf M).
	\end{equation}
	The proof can be found in \appendexref{sec::proofs}.
\end{lemma}
\subsection{Constitutive equations}
In the present paper we assume linear constitutive equations of the form
\begin{equation}\label{eq::constitutionMoment}
\mathbf M = -\frac{t^3}{12}\,\mathcal E : \pmb \rho, 
\end{equation}
and
\begin{equation}\label{eq::constitutionNormal}
\mathbf N = t\,\mathcal E : \pmb\gamma - \mathbf H\cdot\mathbf M.
\end{equation}
The fourth order elasticity tensor $\mathcal E$ is given by
\begin{equation}\label{eq::constitution}
\begin{aligned}
\mathcal E = \lambda (\mathbf P \otimes \mathbf P) + 2 \mu \mathcal P^s,\qquad \text{with} \quad \lambda =\frac{4\bar\lambda\mu}{\bar\lambda+2\mu}, 
\end{aligned}
\end{equation}
where $\bar\lambda$ and $\mu$ are the Lam\'e constants of the elastic material constituting the shell and $\mathcal P^s$ the symmetric part of the tangential fourth order identity tensor. The Lam\'e constants are related to the Young's modulus $E$ and Poisson's ratio $\nu$ by
\begin{equation}
\begin{aligned}
\bar \lambda = \frac{E\nu}{(1+\nu)(1-2\nu)}, \qquad \mu = \frac{E}{2(1+\nu)}.
\end{aligned}
\end{equation}
The constitutive equations can also be write as
\begin{equation}\label{eq::constitutionAgain}
\begin{aligned}
\mathbf N &= \bar {\mathbf N} - \mathbf H \cdot \mathbf M, \\
\mathbf M &= -\frac{t^3}{12}(\lambda \, \projector \,\text{tr}\,\pmb \rho + 2\mu \,\pmb \rho), \\
\bar {\mathbf N} &= t(\lambda\, \projector \,\text{tr}\,\pmb \gamma + 2\mu \,\pmb \gamma).
\end{aligned}
\end{equation}
We remark that with \eqref{eq::constitutionNormal} the condition \eqref{eq::symetryMoment} is fulfilled identically.
\subsection{Weak form of the governing equations}
The weak form of the governing equations is given by
\begin{equation}\label{eq:weakForm}
\begin{aligned}
t\int_\Surface  & \pmb\gamma(\mathbf v):\mathcal E : \pmb\gamma(\mathbf u) \ddx + \frac{t^3}{12} \int_\Surface  \pmb\rho(\mathbf v):\mathcal E : \pmb\rho(\mathbf u)  \ddx = \int_\Surface \mathbf v \cdot \mathbf b \ddx\\ 
& + \int_{\Gamma_{N_i}} v_i \, N^N_i \dd s_x - \int_{\Gamma_{N_t}} \nabla_\Surface (\RealNormalvector\cdot\mathbf v) \cdot \mathbf t \; M^N_t \dd s_x - \int_{\Gamma_{N_\mu}} \nabla_\Surface (\RealNormalvector\cdot\mathbf v) \cdot \pmb \mu \; M^N_\mu \dd s_x,
\end{aligned}
\end{equation}
where $\mathbf v$ are appropriate test functions, see \Cref{sec::derivation_weak}. The boundary conditions which can be prescribed are given by,
\begin{equation}
\begin{aligned}
\mathbf u \cdot \mathbf e_i &=  u^D_i  \quad &\text{or} \quad &&
\mathbf e_i\cdot(\bar{\mathbf N} + \RealNormalvector \otimes \mathbf S) \cdot \pmb \mu &= N^N_i, \\
\nabla_\Surface(\RealNormalvector\cdot \mathbf u)\cdot  \mathbf t &=  \omega_{t} \quad &\text{or} \quad && 
\mathbf t \cdot \mathbf M \cdot \pmb \mu &=  M^N_t, \\
\nabla_\Surface(\RealNormalvector\cdot \mathbf u) \cdot \pmb \mu & =  \omega_{\mu} \quad &\text{or} \quad && 
\pmb \mu \cdot \mathbf M \cdot \pmb \mu &= M^N_\mu,
\end{aligned}
\end{equation}
where $u^D_i$ is a given displacement in the direction of $\mathbf e_i$, $\omega_{t}$ and $\omega_{\mu}$ are given rotations, $N^N_i$ is a given force in the direction of $\mathbf e_i$, $M^N_t$ and $M^N_\mu$ are given moments. 
If $\mathbf u$ is prescribed on the boundary, the derivative along the boundary $d_\mathbf t (\RealNormalvector\cdot\mathbf u ) = \nabla_\Surface (\RealNormalvector\cdot\ureal)\cdot\mathbf t$ is also prescribed \cite{basar1985}. Thus, in this case, only the normal derivative $d_{\pmb \mu} (\RealNormalvector\cdot\mathbf u ) = \nabla_\Surface (\RealNormalvector\cdot\ureal)\cdot\pmb \mu$ can be independently prescribed by $\omega_{\mu}$.

\stuff{
Definition of the divergence operator
\begin{equation}
\text{div} \mathbf F = \n \cdot \text{curl}(\n\times\mathbf F)
\end{equation}	

The divergence product rule
\begin{equation}
\text{div}(\mathbf T \cdot \mathbf v) = \text{div}(\mathbf T^\top)\cdot\mathbf v + \mathbf T^\top :\nabla_\Omega \mathbf v
\end{equation}
\begin{equation}
\text{div}(\mathbf v\cdot\mathbf T) = \mathbf v\cdot\text{div}(\mathbf T) + \nabla_\Omega \mathbf v:\mathbf T^\top
\end{equation}
\begin{equation}
\text{div}(\mathbf v\times\mathbf T) = \mathbf v\times\text{div}(\mathbf T) + \nabla_\Omega \mathbf v\ctimes\mathbf T^\top 
\end{equation}
with $(\mathbf a \otimes \mathbf b) \ctimes(\mathbf c \otimes \mathbf d) = (\mathbf b\cdot\mathbf c) (\mathbf a\times\mathbf d)$.
\begin{equation}
\text{div}(\mathbf u\otimes\mathbf v) = \mathbf u \;\text{div}(\mathbf v) + \nabla_\Omega \mathbf u\cdot\mathbf v 
\end{equation}
\begin{theorem}\label{the::adjoint}
	\begin{equation}
	\mathbf T :(\mathbf M\cdot\mathbf H) = (\mathbf H \cdot\mathbf T) : \mathbf M
	\end{equation}
\end{theorem}
} 
\subsection{Equivalence to the Koiter model}
In this section the equivalence of the classical linear Koiter model \cite{koiter1966nonlinear,ciarlet2006} and \eqref{eq:weakForm} is outlined. To this end, we define the energy functional 
\begin{equation}\label{eq::energy}
\begin{aligned}
\mathcal E(\mathbf v) = \frac{1}{2} &\left[\int_{\Surface} t \,\pmb\gamma(\mathbf v) :\mathcal E:\pmb\gamma(\mathbf v)+ \frac{t^3}{12}\pmb\rho(\mathbf v):\mathcal E:\pmb\rho(\mathbf v) \ddx\right]  - \int_{\Surface} \bodyforce \cdot \mathbf v \ddx \\&- \int_{\Gamma_{N_i}} v_i \, N^N_i \dd s_x + \int_{\Gamma_{N_t}} \nabla_\Surface (\RealNormalvector\cdot\mathbf v) \cdot \mathbf t \; M^N_t \dd s_x + \int_{\Gamma_{N_\mu}} \nabla_\Surface (\RealNormalvector\cdot\mathbf v) \cdot \pmb \mu \; M^N_\mu \dd s_x.
\end{aligned}
\end{equation}
Then, the linear Koiter shell model reads: Find $\displacement \in \mathcal V$ such that 
\begin{equation}
\begin{aligned}
\mathcal E(\displacement) = \inf_{\mathbf v \in \mathcal V} \mathcal E(\mathbf v).
\end{aligned}
\end{equation}
Since the variational equations of \eqref{eq::energy} are the equations given in \eqref{eq:weakForm} we have established the equivalence. Therefore, we conclude that the Koiter model proposed out of purely mechanical and geometrical intuitions can be derived from first principles of continuum mechanics.

%% file: s_c1TraceFEM.tex
For the discretization of the weak form \eqref{eq:weakForm} we propose a $C^1$-continuity version of the TraceFEM. Following the TraceFEM concept the ansatz space on the surface is defined as the restriction (trace) of an outer ansatz space defined on a background mesh. We label the present method also a Finite-Cell method because we use as a background mesh a Cartesian grid. On this structured grid we are able to construct $C^1$-continuity shape functions by the tensor product of univariate cubic Hermite form functions. Locally, they are defined on the on the unit interval by (see \Cref{fig::localShape})
\begin{equation}\label{eq::hermite}
\begin{aligned}
\varphi_1(\xi) &= 1 + \xi^2(2\xi-3), \\
\varphi_2(\xi) &= \xi(\xi(\xi-2)+1), \\
\varphi_3(\xi) &= -\xi^2(2\xi-3), \\
\varphi_4(\xi) &= \xi^2(\xi-1). 
\end{aligned}
\end{equation}
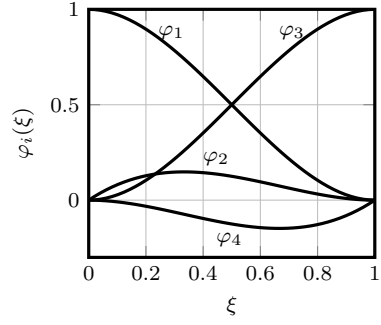
\begin{figure}[h]
	\centering
	\begin{tikzpicture}[]
	\pgfplotsset{%
		width=.45\textwidth,
		height=0.41\textwidth
	}
	\begin{axis}[very thick,xmin=0,xmax=1,ymin=-0.3,ymax=1,grid=both,xlabel={$ \xi$}, ylabel={$\varphi_i(\xi)$}]
	\addplot [domain=0:1,samples=50]({x},{2*x^3-3*x^2+1}) node[above,pos=0.25]{$\varphi_1$};
	\addplot [domain=0:1,samples=50]({x},{x^3-2*x^2+x}) node[above,pos=0.46]{$\varphi_2$}; ;  
	\addplot [domain=0:1,samples=50]({x},{-2*x^3+3*x^2}) node[above,pos=0.75]{$\varphi_3$};; 
	\addplot [domain=0:1,samples=50]({x},{x^3-x^2}) node[below,pos=0.48]{$\varphi_4$};; 
	\end{axis}
	\end{tikzpicture}
	\caption{Local univariate cubic Hermite form functions}
	\label{fig::localShape}
\end{figure}  
The local functions \eqref{eq::hermite} are pieced together to global $C^1$ shape functions $N^i(\pos)$ by the standard finite element procedure of relating local degrees of freedom and global degrees of freedom. The degrees of freedom are the values and the first derivatives at the vertices of the background grid. Thus, at each vertex we have $2^3=8$ degrees of freedom for the discretization of a scalar field and $24$ degrees of freedom for the vector-valued displacement field. Thus, on a background cell we have $64$ local form functions and $192$ local degrees of freedom for the displacement field. We denote the vector-valued finite element space on the background grid by $V_h$. Then, the discrete problem is given by: Find $\displacement_h \in V_h$ such that 
\begin{subequations} \label{eq::discreteProblem}
\begin{equation}\label{eq::discreteProblemA}
M_h(\mathbf u_h,\mathbf w_h) = f_D(\mathbf w_h),
\end{equation}
holds for all $\mathbf w_h\in V_h$  and 
\begin{equation}\label{eq::discreteProblemB}
K_h(\mathbf u_h,\mathbf v_h) = f_N(\mathbf v_h), 
\end{equation}
\end{subequations}
holds for all $\mathbf v_h$ with $M_h(\mathbf v_h, \mathbf w_h)=0$ for all $\mathbf w_h \in V_h$.  
The linear and bilinear forms  are given by
\begin{subequations}\label{eq::forms}
\begin{align} \label{eq::formsA}
M_h(\mathbf u_h,\mathbf v_h) &= \int_{\Gamma_{D_i}} \mathbf v_h  \cdot \mathbf u_h \ddsx + \int_{\Gamma_{D_c}} [\nabla_\Surface (\RealNormalvector\cdot\mathbf v_h) \cdot \pmb \mu] [\nabla_\Surface(\RealNormalvector\cdot\mathbf u_h) \cdot \pmb \mu] \ddsx, \\
f_D(\mathbf v_h) &= \int_{\Gamma_{D_i}} \mathbf v_h  \cdot \mathbf u_{D_i} \ddsx + \int_{\Gamma_{D_c}} [\nabla_\Surface (\RealNormalvector\cdot\mathbf v_h) \cdot \pmb \mu]  \;  u_{D_c} \ddsx, \label{eq::formsB} \\
K_h(\mathbf u_h,\mathbf v_h) &= t\int_\Surface   \pmb\gamma(\mathbf v_h):\mathcal E : \pmb\gamma(\mathbf u_h) \ddx + \frac{t^3}{12} \int_\Surface  \pmb\rho(\mathbf v_h):\mathcal E : \pmb\rho(\mathbf u_h) \ddx, \label{eq::formsC}   \\
f_N(\mathbf v_h) &=\int_\Surface \mathbf v_h \cdot \mathbf b \ddx + \int_{\Gamma_{N_i}} v_{h,i} \, N^N_i \dd s_x \nonumber \\ &\qquad- \int_{\Gamma_{N_t}} \nabla_\Surface (\RealNormalvector\cdot\mathbf v) \cdot \mathbf t \; M^N_t \dd s_x - \int_{\Gamma_{N_\mu}} \nabla_\Surface (\RealNormalvector\cdot\mathbf v) \cdot \pmb \mu \; M^N_\mu \dd s_x. \label{eq::formsD}
\end{align}
\end{subequations}
In \eqref{eq::discreteProblemA} the solution $\displacement_h$ gets fixed to prescribed values at the Dirichlet boundary. However, the matrix $M_h$ (and also $K_h$) is singular by definition because of two reason. First, in order to avoid further notation, we take $\mathbf w_h\in V_h$ resulting in zero rows and columns, which can be easily identified. Secondly, since we define the shape functions by restriction of the shape functions defined on the background mesh, they are not linear independent. In the standard setting of a fitted finite element method the respective degrees of freedom in \eqref{eq::discreteProblemA} are easy to identify and can be determined by interpolation. Here, in the case of an unfitted method a linear independent basis is not known explicitly in general and have to be determine. In \eqref{eq::discreteProblemB} the test functions $\mathbf v_h$ are restricted to the null-space of $M_h$. 

In the discrete method the integrals in \eqref{eq::forms} are evaluated by quadrature, which is described in the next section.
\subsection{Integral evaluation}
In order to evaluate the surface and line integrals in \eqref{eq::forms} we use the quadrature schema developed in \cite{saye2015high}. Here, we outline only the main ingredients and refer to \cite{saye2015high} for technical details. Following the standard finite element procedure the integrals are evaluated by summing up background cell (face) contributions where the shape functions are smooth. The individual contributions are evaluated by Gaussian quadrature. The main idea from \cite{saye2015high} is to subdivide the background cells (faces) until it is possible to convert the implicitly defined geometry into
the graph of an implicitly defined height function. Then, a recursive algorithm which requires only one-dimensional root finding and one-dimensional Gaussian quadrature can be set up. In order to chose suitable height function directions we have to be able to ensure the monotonicity of the level-set function in that direction. This can be done by showing that the derivative in that direction is uniform in sign, \ie place bounds on the values attainable by the derivative. In contrast to \cite{saye2015high} we use interval arithmetic for this task. 

\subsection{Solution strategies}\label{sec::solutionMethods}
In order to solve \eqref{eq::discreteProblem} we consider the three methods,
\begin{itemize}
	\item Null-space method, 
	\item Penalty method, and
	\item Lagrange multiplier method.
\end{itemize}
In the null-space method we first solve 
\begin{equation}\label{eq::nullA}
M_h u^D_h = f_D,
\end{equation}
and compute the null-space of $M_h$. We denote the null-space basis by $Z_h$. In a next step the solution $u^0_h$ of 
\begin{equation}\label{eq::nullB}
(Z_h^\top K_h Z_h) u^0_h = Z_h^\top (f_N - K_h u^D_h)
\end{equation}
is computed. The overall solution is then given by
\begin{equation}
u_h = Z_h u^0_h + u^D_h.
\end{equation}
In the penalty method we solve a system of linear equations of the form
\begin{equation}\label{eq::penalty}
(K_h + \alpha M_h) \, u_h = f_N + \alpha f_D,
\end{equation}
with the penalty parameter $\alpha > 0$. In the Lagrange multiplier method we solve the system of linear equation
\begin{equation}\label{eq::lagrange}
\begin{bmatrix}
K_h & M_h \\
M_h & 0 
\end{bmatrix} \begin{bmatrix}
u_h \\ \lambda_h
\end{bmatrix} = \begin{bmatrix} f_N  \\  f_D \end{bmatrix},
\end{equation}
with the Lagrange multipliers $\lambda_h$.

\subsection{Implementation}
The proposed method has been implemented in Matlab. Within the method the exact level-set function $\phi(\pos)$ is used. For the evaluation of the surface normal vector \eqref{eq::real_normal_vector} and the Weingarten map \eqref{eq::real_weingarten}, the first and second order derivatives of the level-set function are necessary. In the implementation we use symbolic differentiation of $\phi(\pos)$ to provide these derivatives. 

In the present work, we have not used any stabilization term which is added to the weak form. Therefore, in each system of equations \labelcref{eq::penalty,eq::lagrange,eq::nullA,eq::nullB} the system matrix is singular by definition. One strategy to would be to add stabilization terms to the bilinear forms \eqref{eq::discreteProblemB} and \eqref{eq::discreteProblemB}. We refer to \cite{Olshanskii_Reusken_2017} for an overview of different possibilities.  Although such a stabilization can be designed in such a way that the convergence order of the method is not been altered, a stabilization decreases the accuracy of the method. Therefore, a strategy is investigated, where no stabilization is necessary. However, we have observed in numerical experiments that the Matlab  backslash operator does not give satisfactory results. Due to this reason, we use the direct solver suitable for under-determined linear equation systems from the  SuiteSparse\footnote{http://faculty.cse.tamu.edu/davis/suitesparse.html} project.

%% file: s_numerical_results.tex
In this section, numerical results are presented. First, we verify the implementation of the proposed method against  exact manufactured solutions. Secondly, we demonstrate that the method works by testing it with two benchmark problems (one cylindrical shell and one spherical shell) of the well-known shell obstacle course \cite{belytschko1985}. Finally, in a fourth example, a complex shell geometry is investigated.
\subsection{Verification example}
The implementation of the proposed method is verified. We have successfully run the method on various geometries, displacement fields and boundary condition combinations, but present only the results of two configurations. In all verification examples we used $B=[-0.4,0.8]\times[0,1]\times[-0.3,1]$ and a shell thickness $t=0.25$. The material parameters are $E=10$ and $\nu=0.4$. The considered geometries are defined by the zero level-set of the functions
\begin{subequations} 
	\begin{align}
		\phi_1(x,y,z) &= x + z - 0.7,  \label{eq::verify_geoA} \\
		\phi_2(x,y,z) &= 4x^2 + 0.25y^2 + z^2 - 0.7, \label{eq::verify_geoB}
	\end{align}
\end{subequations}
\begin{figure}[h t]
	\hfill
	\subfloat[$\phi_1(x,y,z) = x + z - 0.7$]
	{\includegraphics[width=0.4\textwidth]{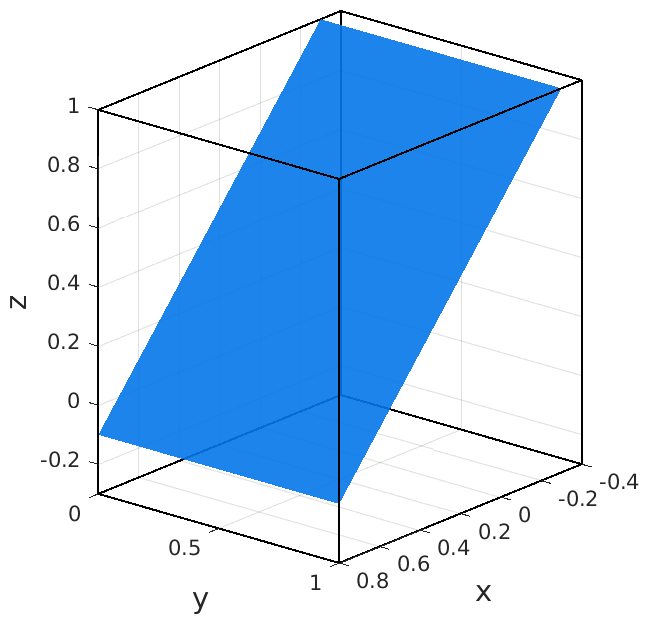}}\hfill
	\subfloat[$\phi_2(x,y,z) =4x^2 + 0.25y^2 + z^2 - 0.7$ ]
	{\includegraphics[width=0.4\textwidth]{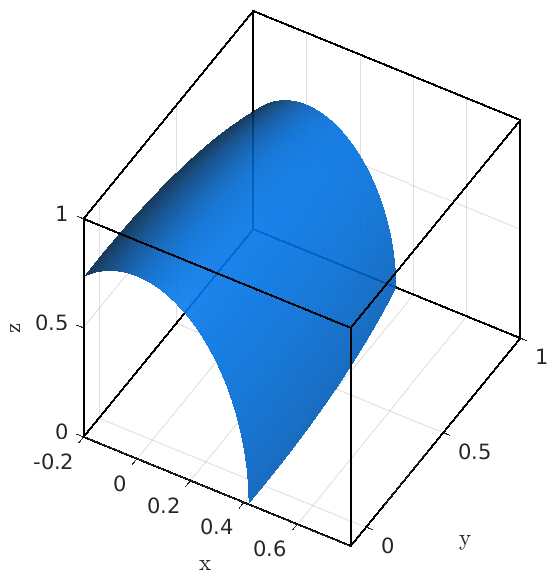}}
	\hfill
	\caption{Problem geometries of the verification}
	\label{fig::verify_geometry}	
\end{figure}
and are illustrated in \Cref{fig::verify_geometry}. The level-set function \eqref{eq::verify_geoA} implies a flat geometry, whereas \eqref{eq::verify_geoB} implies a surface with varying curvature. 

As manufactured solutions we consider the two displacement fields
\begin{subequations}\label{eq::verifySolution}
  	\begin{align}
  	\displacement^{ex}_1(x,y,z) &= x^3\mathbf e_x +y \,x^3\mathbf e_y +(xzy^2+x(x-1)y(y-1))\mathbf e_z, \\
  	\displacement^{ex}_2(x,y,z) &= \sin(16y)\cos(16x\,z)\mathbf e_x+\cos(16x\,y\,z)\mathbf e_y+2\sin(16x\,y\,z)\mathbf e_z.
  	\end{align}
\end{subequations}
Using \eqref{eq::local_force_equilibrium}, we compute symbolically the necessary surface force $\bodyforce$ such that \eqref{eq::verifySolution} is the respective exact solution of the shell problem.
The displacement field $\displacement_1$ is chosen as a third order polynomial such that the solution can be represented exactly in the discrete space, whereas $\displacement_2$ can only be approximated. 

In the following, we study the behavior of the error 
 \begin{equation}
 e = \sqrt{\frac{\int_\Surface (\mathbf u_h - \mathbf u^{ex})^2 \ddx }{\int_\Surface (\mathbf u^{ex})^2 \ddx}}
 \end{equation}
 under uniform mesh refinement for the three different solution methods given in \Cref{sec::solutionMethods}. Furthermore, for comparison, we also consider the Hermite interpolation $\mathbf u^{int}_h$ of the solution on the background grid and the surface $L_2$-projection: Find $\mathbf u^{L_2}_h \in V_h$ such that 
 \begin{equation}
 \int_\Surface (\displacement^{ex}-\mathbf u^{L_2}_h)^2 \ddx \quad \rightarrow \quad \text{min}.
 \end{equation}
 We remark that the Hermite interpolation and the $L_2$-projection are only possible if the solution is known, and is thus only computed for the verification examples in this section. Furthermore, all solutions apart from the Hermite interpolation require the solution of a system of linear equations. 
 
 The numerical results for the flat plate defined in \eqref{eq::verify_geoA} are visualized in \Cref{fig::verify_penalty1,fig::verify_conv1}. The refinement level 0 refers to a single background cell. In \Cref{fig::verify_penalty1} the errors obtained by the penalty method are given for different refinement levels and penalty factors. 
 \begin{figure}[h t]
 	\subfloat[displacement $\mathbf u_1$]{
 		\includegraphics[width=0.49\textwidth]{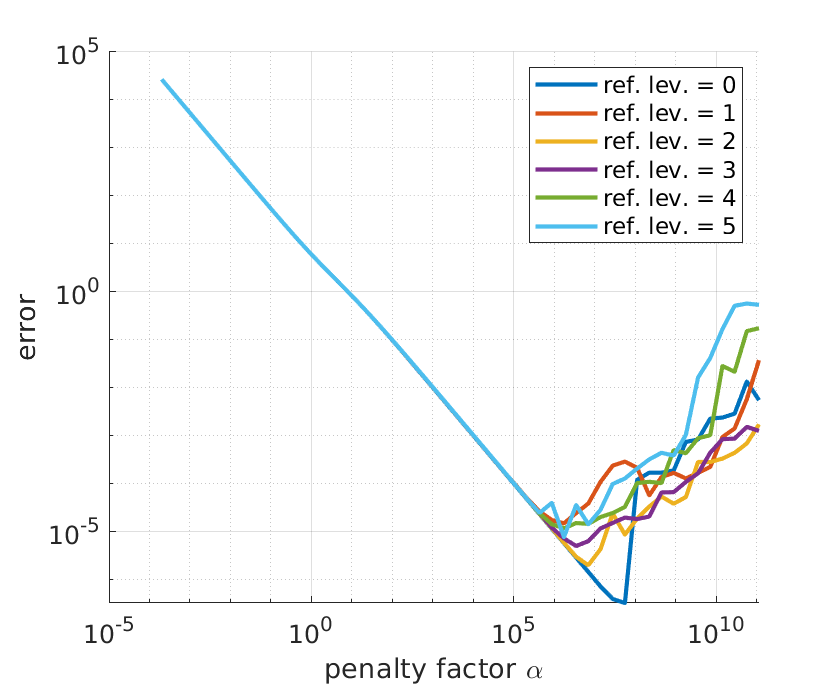}}
 	\subfloat[displacement $\mathbf u_2$]{
 		\includegraphics[width=0.49\textwidth]{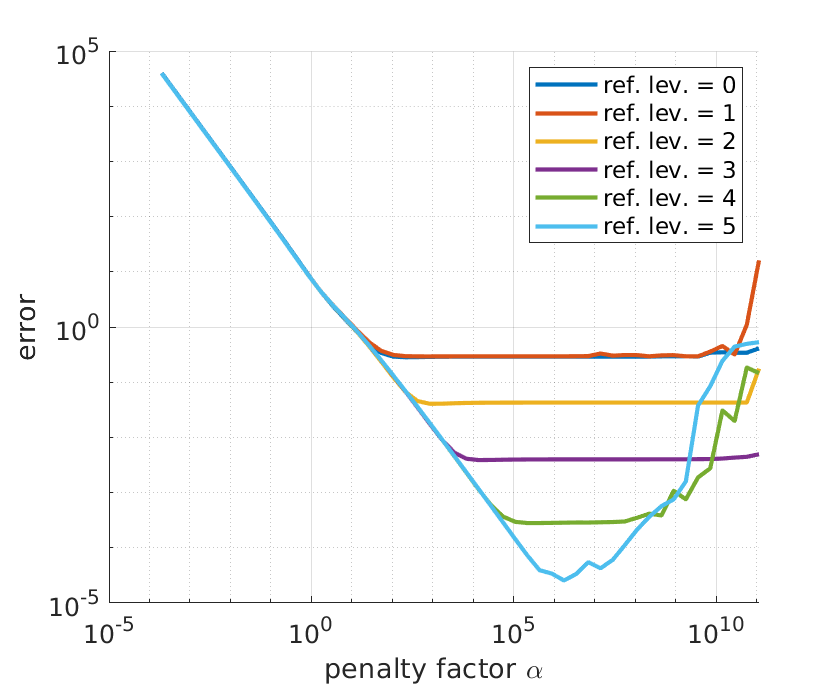}}
 	\caption{Errors for different penalty factors on the geometry induced by $\phi_1$}
 	\label{fig::verify_penalty1}	
 \end{figure}
We remark that for this flat geometry the numerical integration is exact up to round of errors. Therefore, for $\displacement_1$ the sources for errors are the error due to the imposition of the boundary conditions by the penalty method (depending on the penalty factor) and round off errors in the numerical computations. As expected the errors decrease with increasing penalty parameter up to a point where the ill-conditioning of the linear system dominates the error. Since $\displacement_2$ can not be represented exactly in the discrete space an approximation error limits the overall error. 
 \begin{figure}
 	\includegraphics[width=0.99\textwidth]
 	{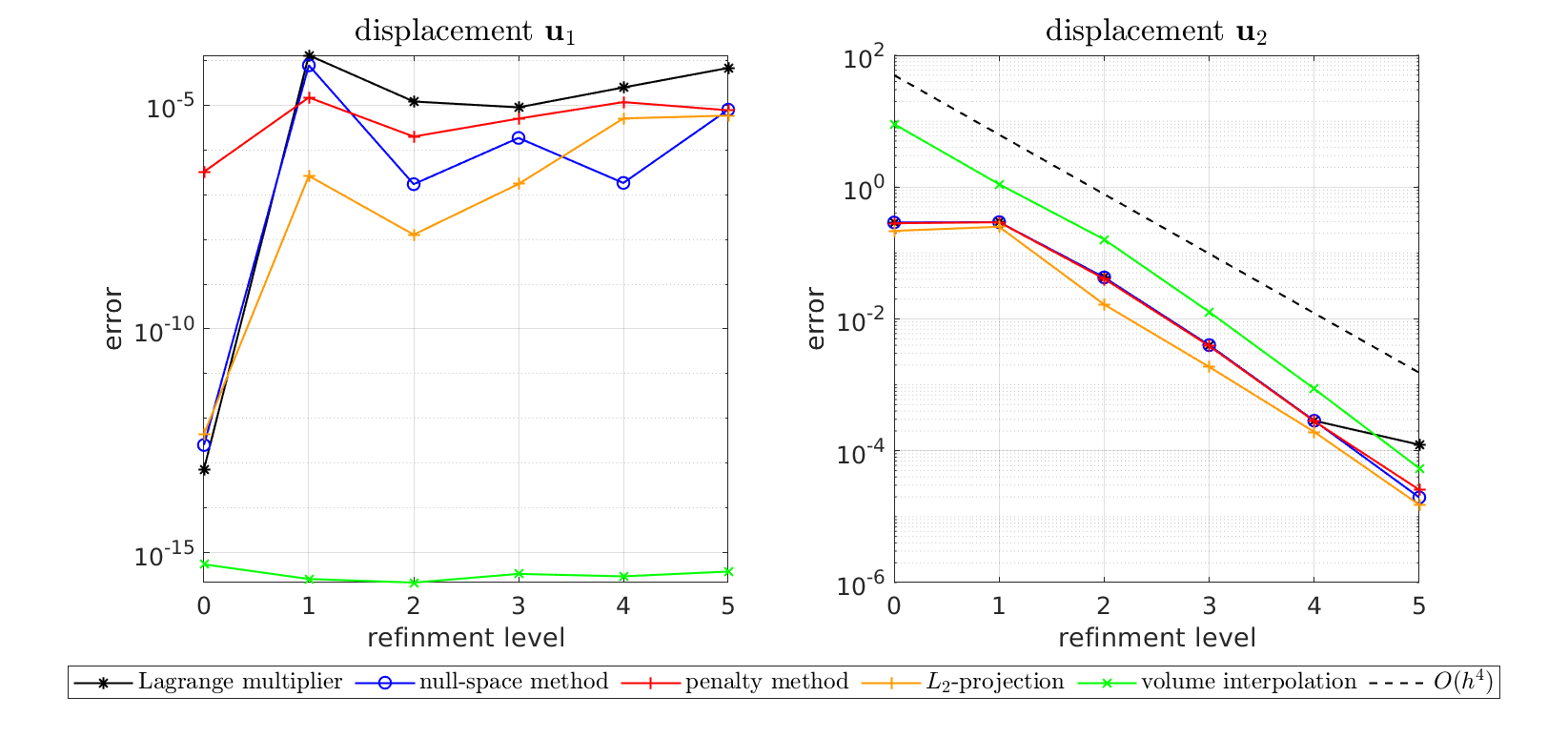}
 	\caption{Results for $\phi_1$ and $\displacement_1$ (left) and $\displacement_2$ (right)}
 	\label{fig::verify_conv1}
 \end{figure}
In \Cref{fig::verify_conv1} the results for the different solution methods are given. For the penalty method we used the lowest errors of the results shown in \Cref{fig::verify_penalty1}. As no system of equation has to be solved for the interpolation on the background grid (\textit{volume interpolation}) the error is around $10^{-16}$ for $\displacement_1$ for all refinement levels. In contrast to this the other results require the solution of a system of equations and therefore the errors are between $10^{-8}$ and $10^{-4}$ due to round off errors. For $\displacement_2$ we observe the convergence of all methods with optimal rate. Here, by definition the $L_2$-projection gives the lowest error, whereas the Hermite interpolation results in the highest error for a fixed refinement level (apart from level 5, where the error due to ill-conditioning of the system of equations dominates).
     
\begin{figure}[h t]
	\subfloat[displacement $\mathbf u_1$]{
		\includegraphics[width=0.49\textwidth]{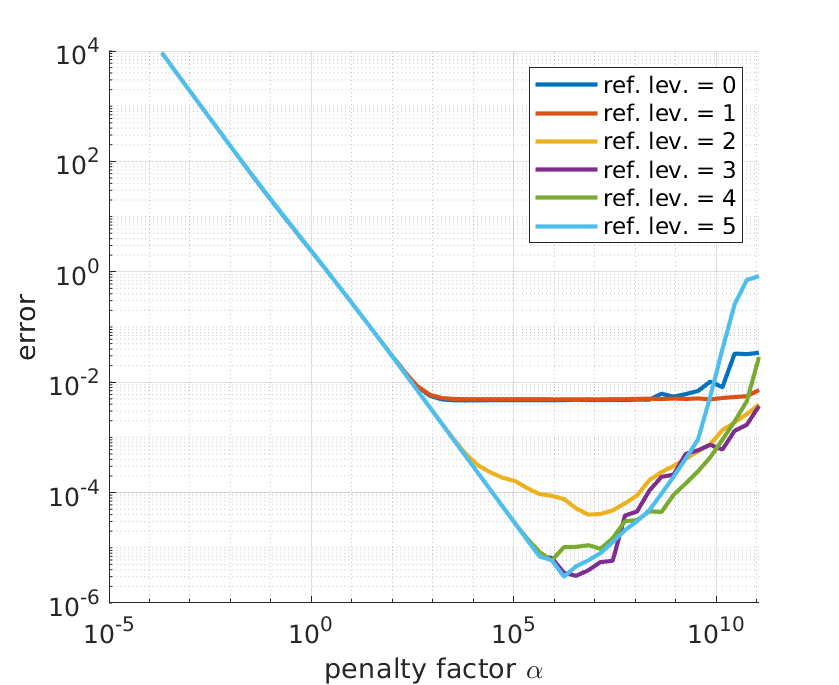}}
	\subfloat[displacement $\mathbf u_2$]{
		\includegraphics[width=0.49\textwidth]{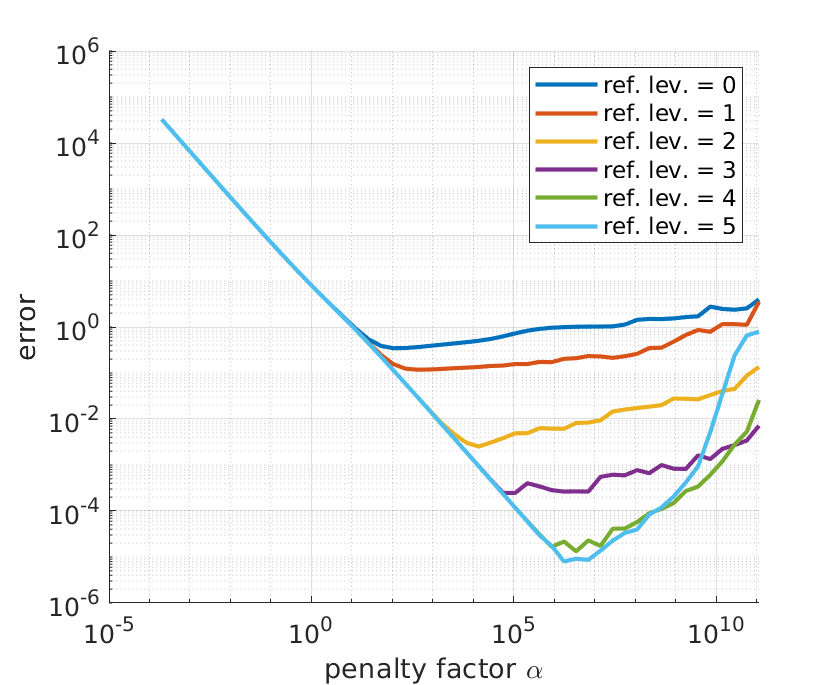}}
	\caption{Errors for different penalty factors on the geometry induced by $\phi_2$}
	\label{fig::verify_penalty2}	
\end{figure}
\begin{figure}
	\includegraphics[width=0.99\textwidth]
	{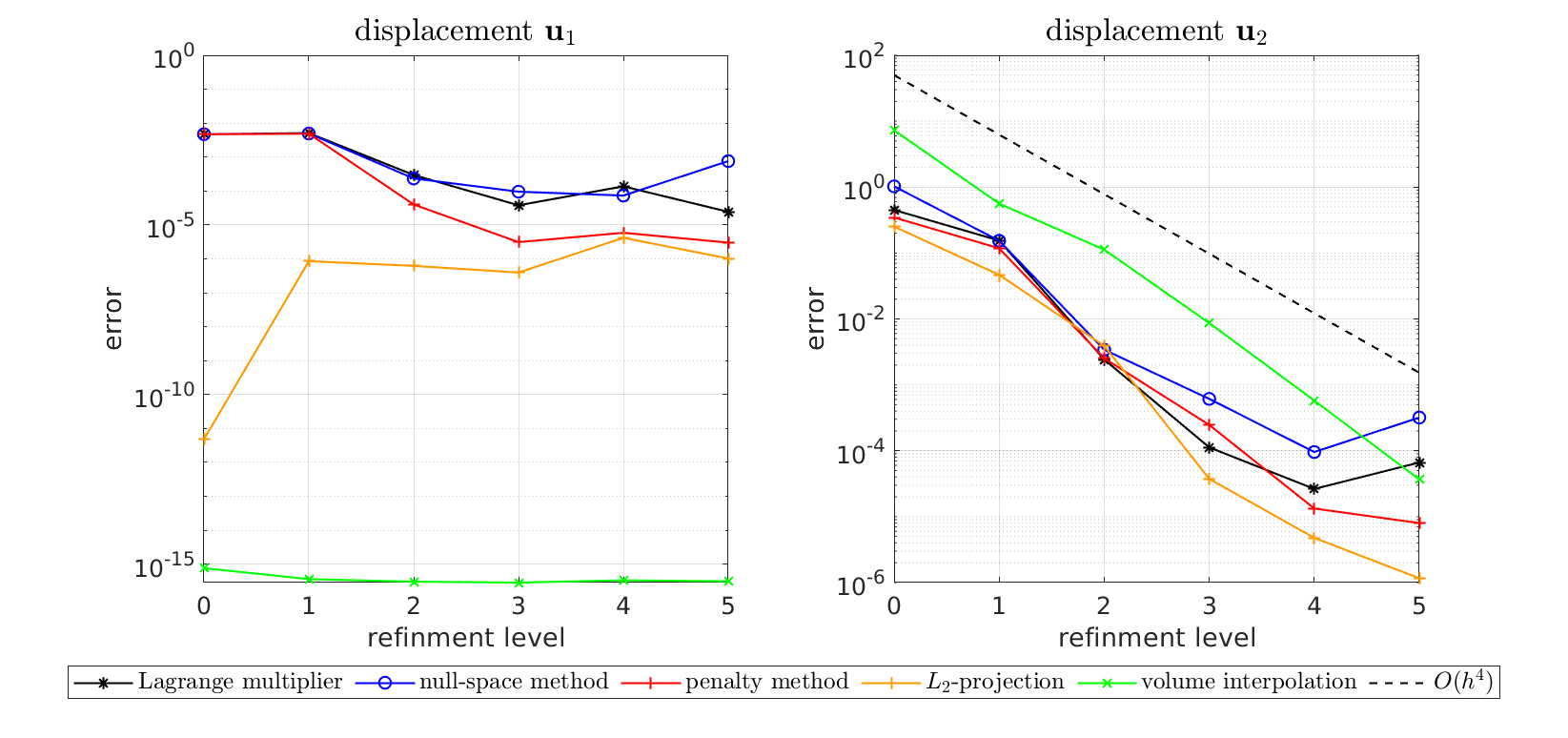}
	\caption{Results for $\phi_2$ and $\displacement_1$ (left) and $\displacement_2$ (right)}
	\label{fig::verify_conv2}
\end{figure}
The numerical results for the part of the ellipse defined in \eqref{eq::verify_geoB} are visualized in \Cref{fig::verify_penalty2,fig::verify_conv2}.  In \Cref{fig::verify_penalty2} the results of the penalty method for different penalty parameters are given. In contrast to the flat geometry, now the numerical integration is not exact, yielding an additional error, which dominates for the three coarsest levels. 

In \Cref{fig::verify_conv2} the errors for the three different solution methods, the Hermite interpolation and the $L_2$-projection are visualized.
Again, for the penalty method we used the lowest errors of the results shown in \Cref{fig::verify_penalty2}. Similar as before, for the displacement field $\displacement_1$ the volume interpolation gives errors around $10^{-16}$, whereas for the other solutions the errors are between $10^{-7}$ and $10^{-2}$ due to round off errors. For displacement field $\displacement_2$ we observe the convergence of all methods. However, due to round off errors the accuracy is limited. Nevertheless, we remark that an relative error level of about $10^{-5}$ is more than sufficient for practical problems. This can also be seen from the visualizations of the solutions obtained with the null-space method for $\displacement_1$ in \Cref{fig::verify_deformation1} and for $\displacement_2$ in \Cref{fig::verify_deformation2}. For the fine levels no difference in the solutions can be seen by eye.
\begin{figure}[h t]
	\subfloat[level $0$]
	{\includegraphics[width=0.33\textwidth]{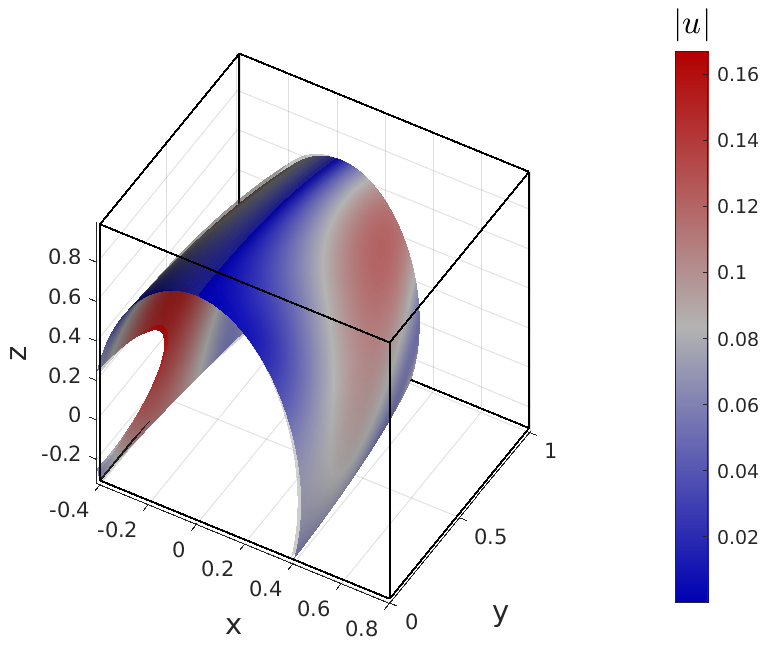}}
	\subfloat[level $1$]
	{\includegraphics[width=0.33\textwidth]{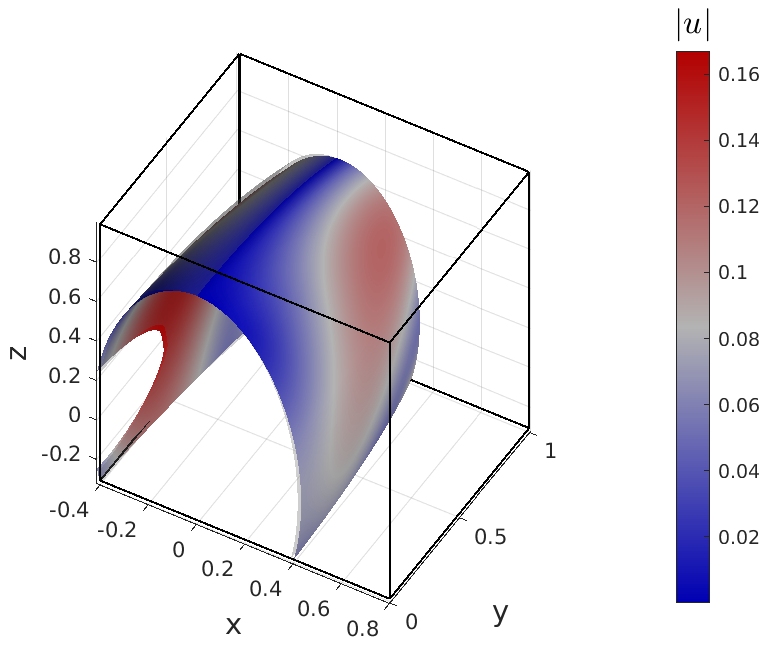}}
	\subfloat[level $2$ ]
	{\includegraphics[width=0.33\textwidth]{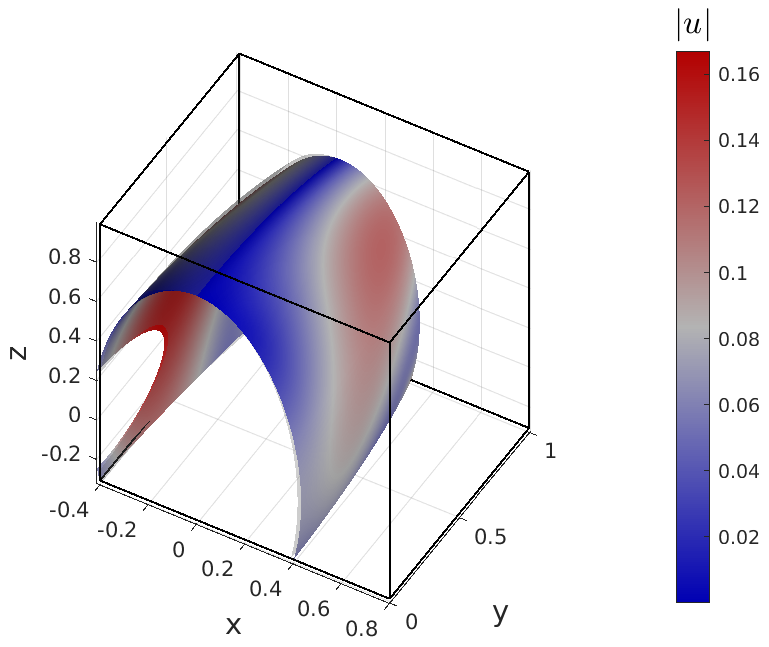}}
	
	\subfloat[level $3$]
	{\includegraphics[width=0.33\textwidth]{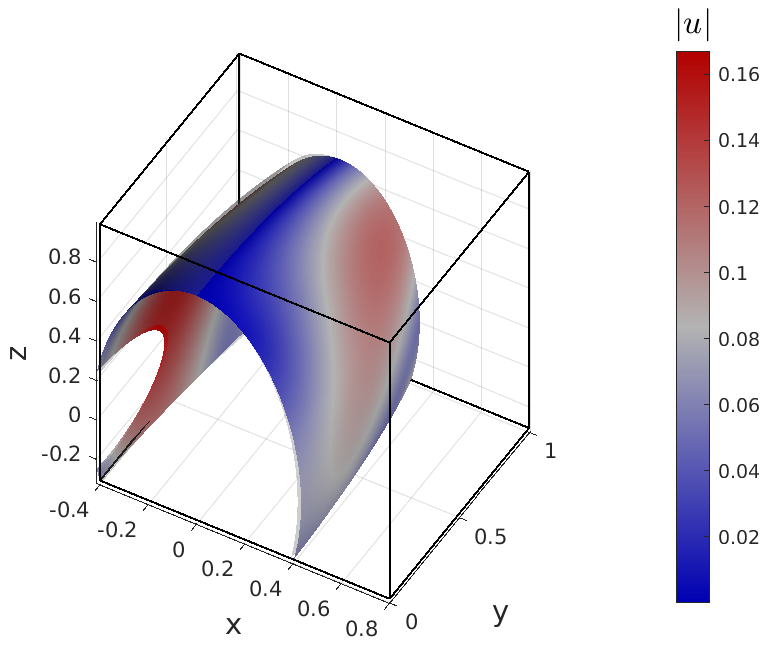}}
	\subfloat[level $4$]
	{\includegraphics[width=0.33\textwidth]{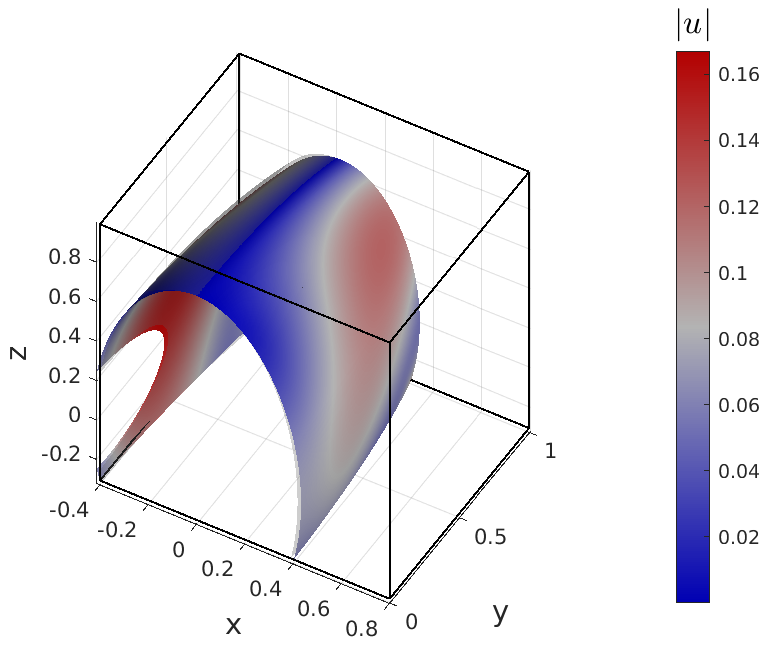}}
	\subfloat[level $5$ ]
	{\includegraphics[width=0.33\textwidth]{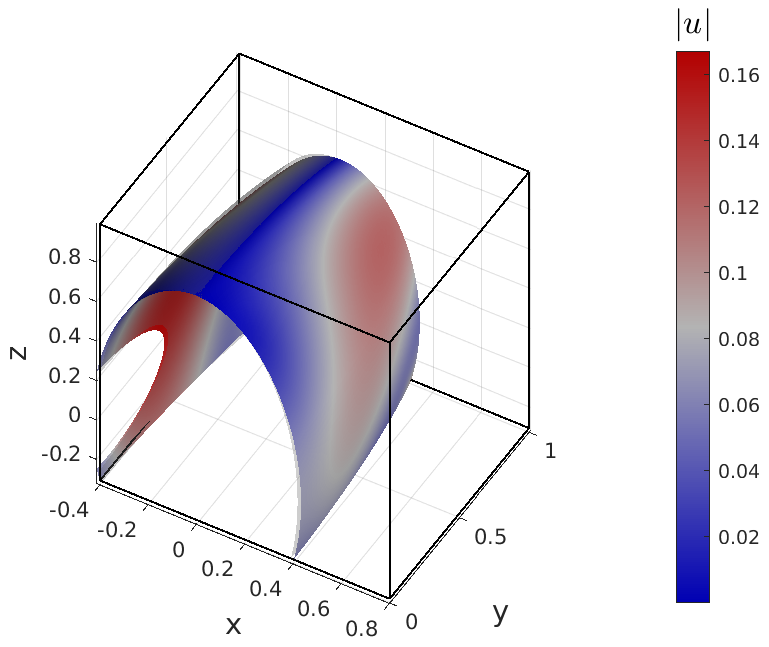}}
	\caption{Visualization of the displacement results for $\phi_2$ and $\mathbf u_1$}
	\label{fig::verify_deformation1}	
\end{figure}
\begin{figure}[h t]
	\subfloat[level $0$]
	{\includegraphics[width=0.33\textwidth]{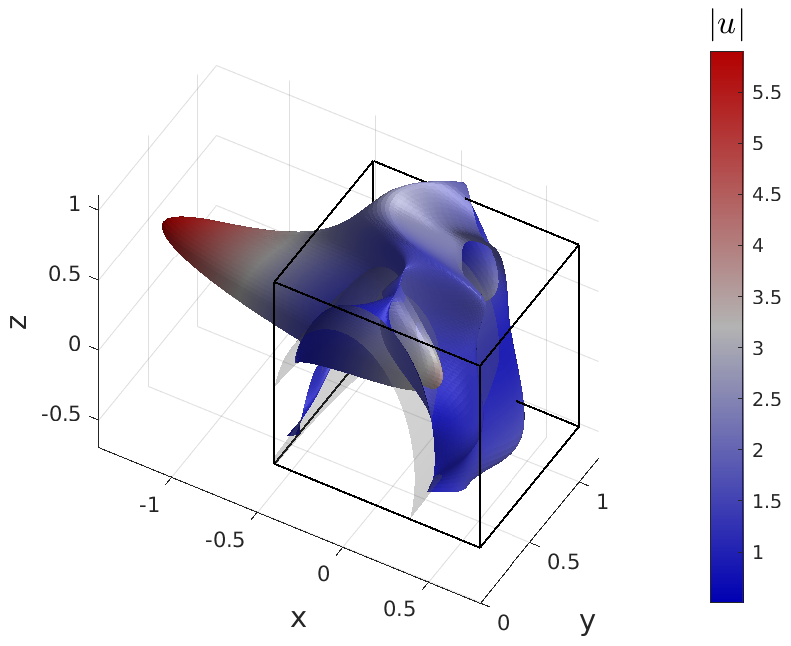}}
	\subfloat[level $1$]
	{\includegraphics[width=0.33\textwidth]{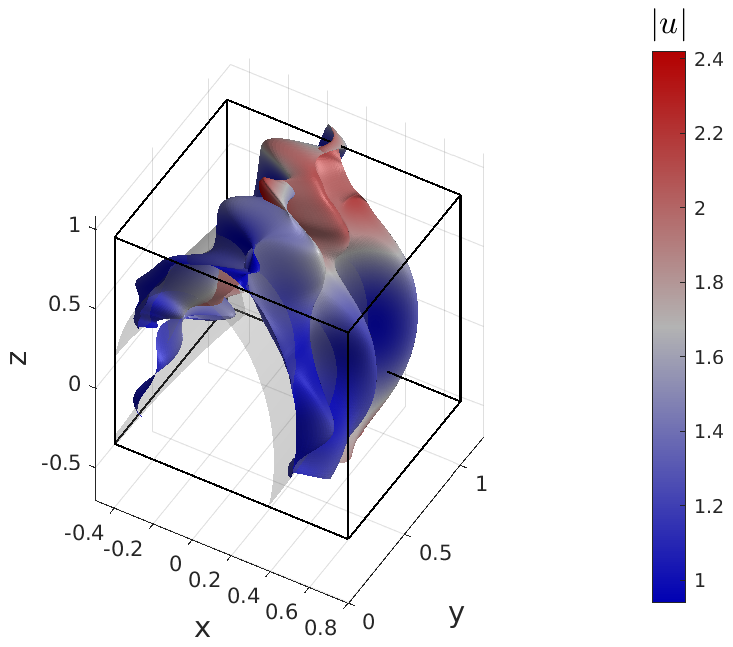}}
	\subfloat[level $2$ ]
	{\includegraphics[width=0.33\textwidth]{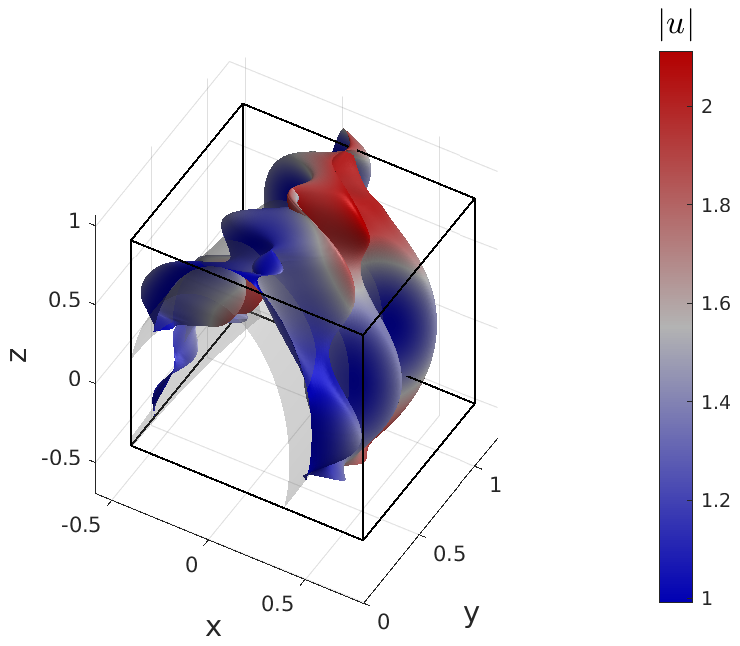}}
	
	\subfloat[level $3$]
	{\includegraphics[width=0.33\textwidth]{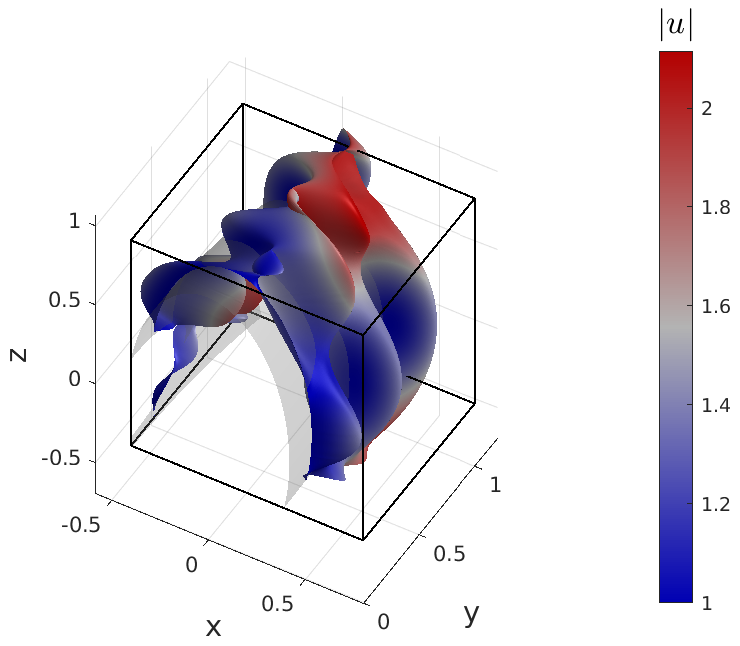}}
	\subfloat[level $4$]
	{\includegraphics[width=0.33\textwidth]{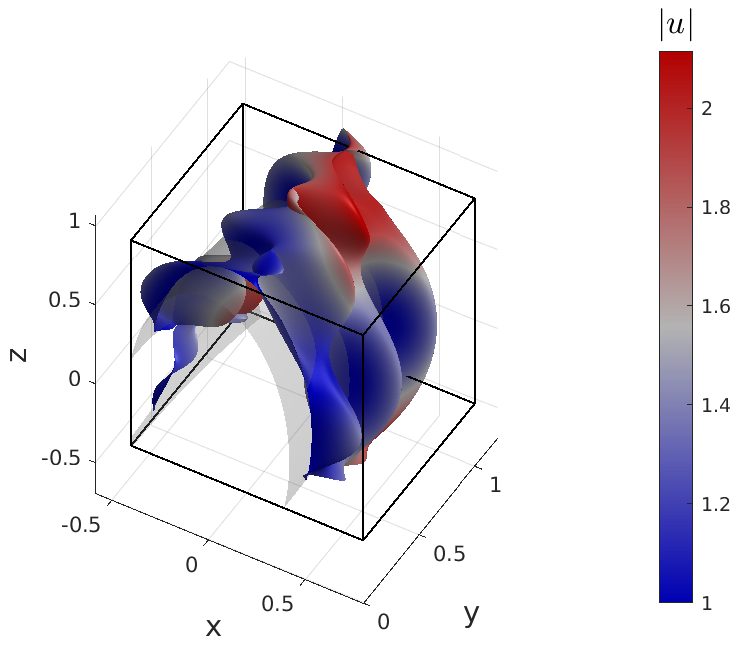}}
	\subfloat[level $5$ ]
	{\includegraphics[width=0.33\textwidth]{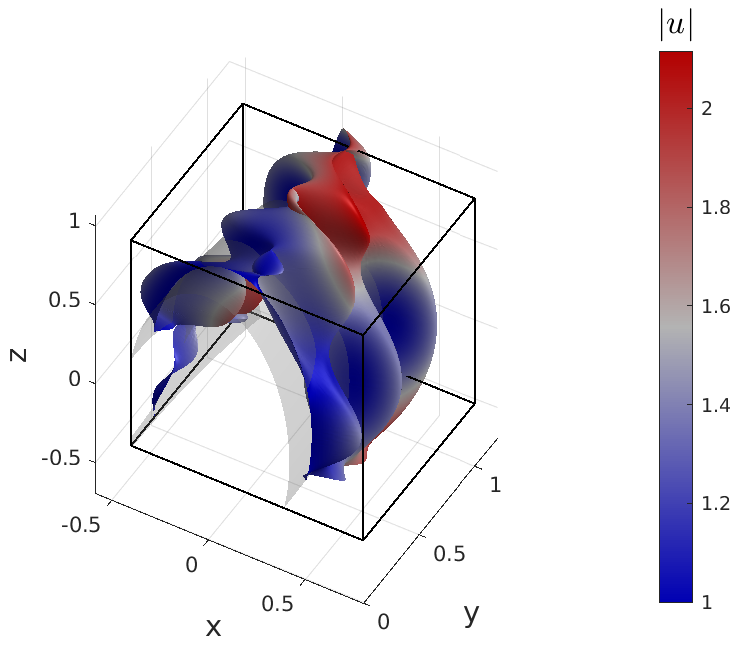}}
	\caption{Visualization of the displacement results for $\phi_2$ and $\mathbf u_2$}
	\label{fig::verify_deformation2}	
\end{figure}   

\subsection{Scordelis-Lo roof}
\begin{figure*}[ h t]
	\begin{minipage}{0.75\textwidth}
		\includegraphics[width=\textwidth]{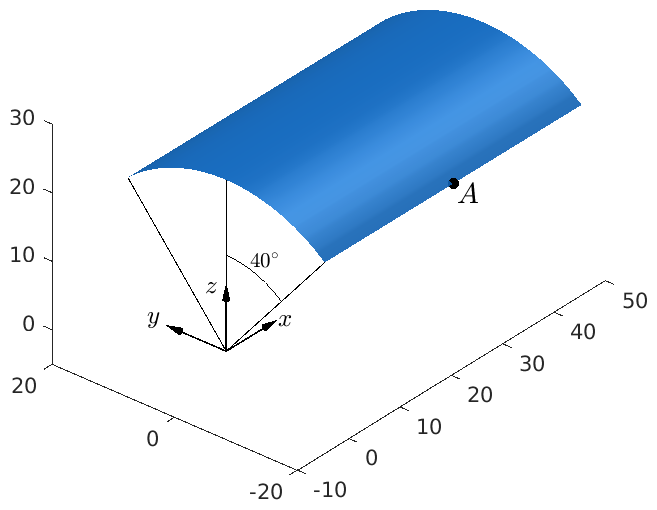}
	\end{minipage}
	\begin{minipage}[h t]{0.19\textwidth}
		\vspace{-3cm}
		\begin{align*}
		E  &=\Nmm{4.32 \cdot 10^8} \\
		\nu&=0  \\
		R &= 25 \\
		L &= 50 \\
		t  &=\um{0.25}
		\end{align*}
	\end{minipage}
	\caption{Problem description of the Scordelis-Lo roof problem}
	\label{fig::ex_para_FEM_roof}
\end{figure*}
We consider the classical Scordelis-Lo roof problem, which is one example from the shell obstacle course \cite{belytschko1985}. It is a popular benchmark test to assess the performance of finite elements regarding complex membrane strain states. The cylindrical roof (radius $r=\um{25}$) is supported by rigid diaphragms at the ends ($x=\um{0}$ and $x=\um{50}$), \ie $u_y=u_z=\um{0}$. The straight edges are free. The geometry and the material parameters are depicted in \Cref{fig::ex_para_FEM_roof}. The structure is subjected to gravity loading with $\mathbf b=-\Nm{90}\,\mathbf e_z $. 
We describe the problem geometry by 
\begin{equation}
	\phi(x,y,z)	= y^2 + z^2 - r^2,
\end{equation}
and $B = [0,50] \times[-r\sin(\frac{40}{180}\pi),r \sin(\frac{40}{180}\pi)]\times[10,31.25]$.	
	
We study the vertical displacement of point $A$, which is located in the middle of one free edge. As a reference solution we use the overkill solution $u_z^A = -0.3006$ from \cite{Bieber_Oesterle_Ramm_Bischoff_2018} obtained by an isogeometric formulation using fifth-order NURBS and a mesh of 48 control points in each direction. The results for different meshes obtained with the presented methods are given in \Cref{tab::roof_dispA}and the deformed geometry is depicteed in \Cref{fig::ex_roof_solution}. We observe that the null-space method and the penalty method are able to reproduce the reference displacement found in literature accurately. However, the results obtained by the Lagrange multiplier method show some instability of the method. The investigation of the origin of these instabilities is topic of further research. 
\begin{figure}[h t]
	\centering
	\includegraphics[width=0.69\textwidth]{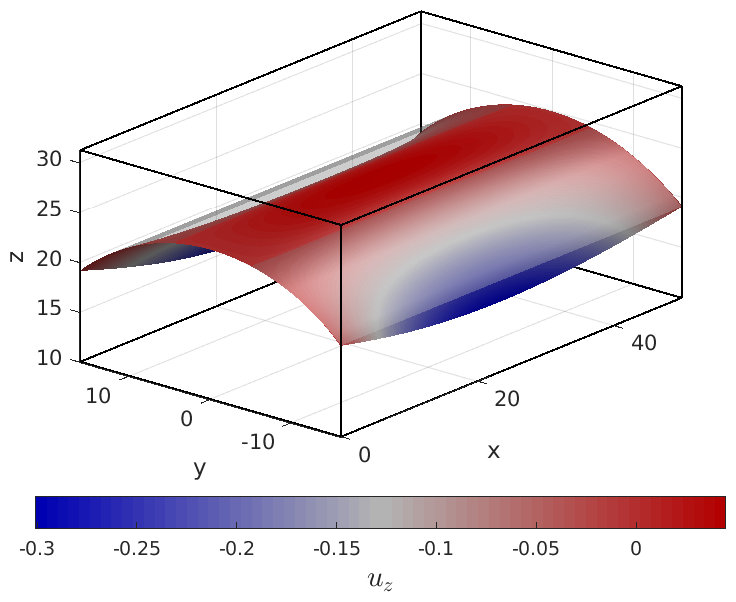}
	\caption{Deformed configuration of the Scordelis-Lo roof (displacements scaled by a factor of 10)}
	\label{fig::ex_roof_solution}
\end{figure}
\begin{table}
	\centering
	\caption{Vertical displacements of the Scordelis-Lo roof at point $A$. Refernece: $u_z^A = -0.3006$.}
	\label{tab::roof_dispA}
	\input{matlab_pics/roof_results.tex}
\end{table} 
\subsection{Pinched hemisphere}
In this example, we consider the pinched hemisphere problem from the shell obstacle course in \cite{belytschko1985}. We describe the spherical mid-surface by
\begin{equation}
\begin{aligned}
\phi(x,y,z) = x^2 + y^2 + z^2 - R^2,
\end{aligned}
\end{equation}    
with $R=\um{10}$ and $B = [-12.5,12.5] \times [-12.5,12.5]\times[0,12.5]$. The material properties and the general problem setup are shown in \Cref{fig::hemisphere}.
\begin{figure*}[h t]
	\begin{minipage}{0.75\textwidth}
		\includegraphics[width=\textwidth]{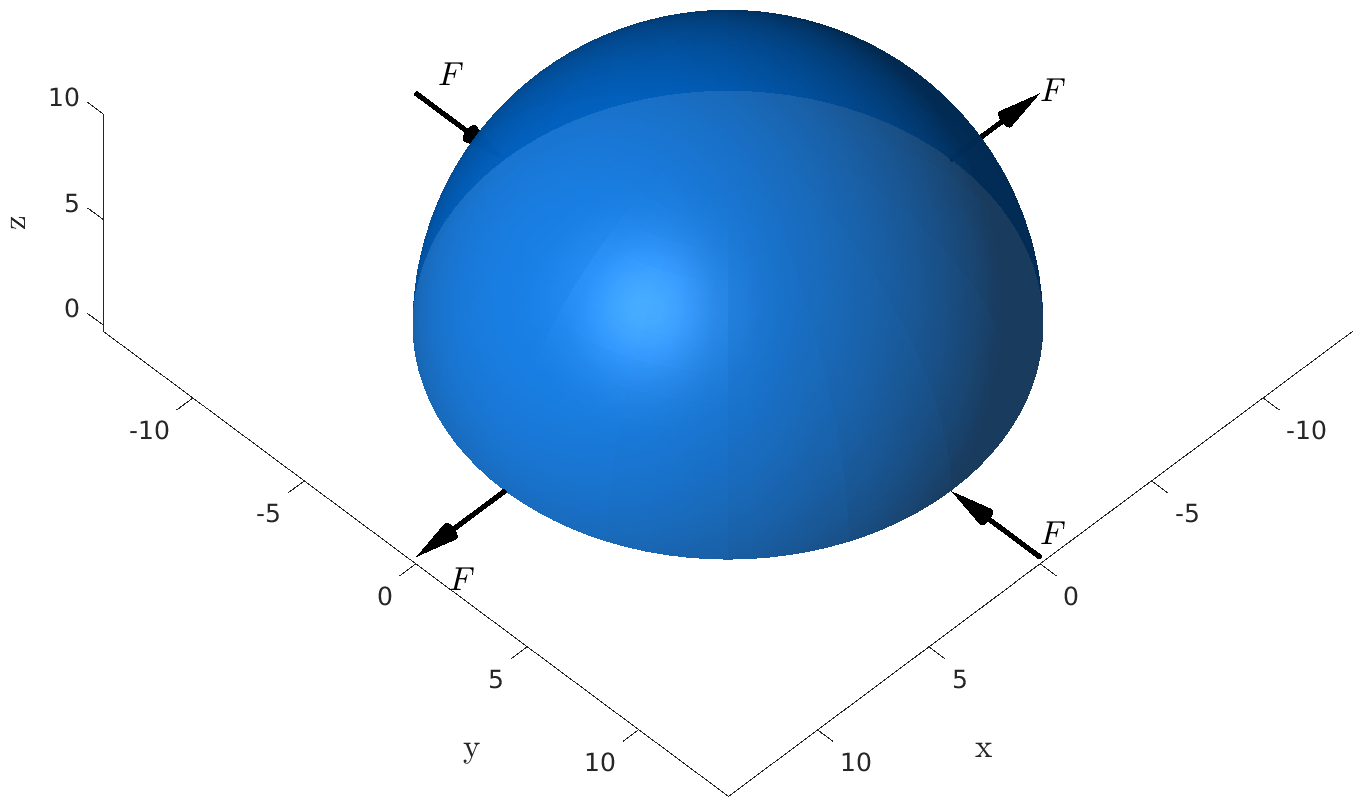}
	\end{minipage}
	\begin{minipage}[h]{0.19\textwidth}
		\vspace{-2cm}
		\begin{align*}
		E  &=\Nmm{6.825 \cdot 10^7} \\
		\nu&=0.3  \\
		R &= 10 \\
		t &=\um{0.04} \\
		F &= 2
		\end{align*}
	\end{minipage}
	\caption{Problem description of the pinched hemisphere problem}
	\label{fig::hemisphere}
\end{figure*}
The edge of the hemisphere is unconstrained and the four radial forces
have alternating signs such that the sum of the applied forces is
zero. We investigate the radial displacement at the loaded points. In
\cite{belytschko1985}, the reference displacement of
$u_r = 0.0924$ is given. The results obtained by the presented methods are given in \Cref{tab::henisphere_dispA} and the deformed configuration is depicted in \Cref{fig::ex_hemisphere}. We observe that here all three methods give nearly the same results. The values obtained for the finest levels are in very good agreement with the reference value found in literature. 
\begin{figure}[h]
	\centering
	\includegraphics[width=0.99\textwidth]{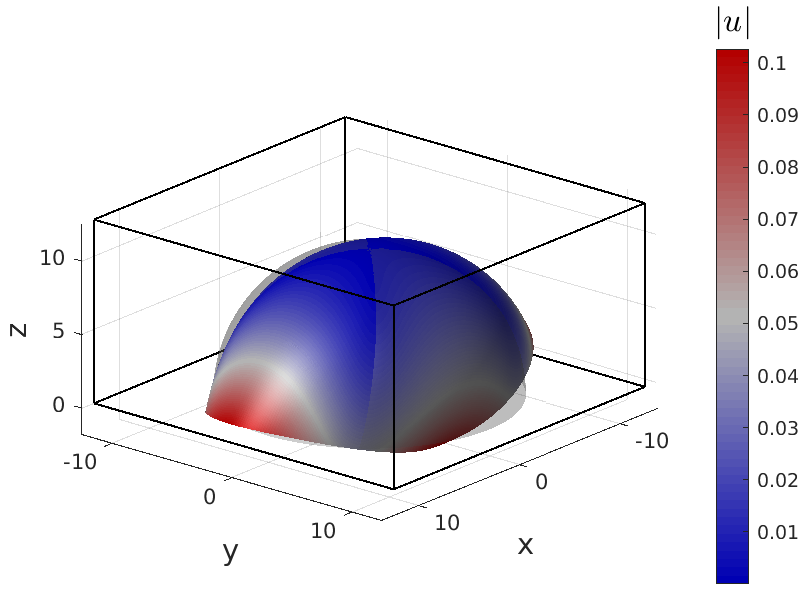}
	\caption{Deformed configuration of the pinched hemisphere (displacements scaled by a factor of 40)}
	\label{fig::ex_hemisphere}
\end{figure}
\begin{table}
	\centering
	\input{matlab_pics/hemisphere_results.tex}
	\caption{Displacements of the pinched hemisphere at one loading point. Refernece: $u_r = 0.0924$.}
	\label{tab::henisphere_dispA}
\end{table}
\subsection{Gyroid}
In this example, we consider the deformation of a shell structure with a complex geometry. The mid-surface is part of a gyroid which is given by the level-set function
\begin{equation}
\phi(x,y,z) = \sin(\pi x)\cos(\pi y)+\sin(\pi y)\cos(\pi z)+\sin(\pi z)\cos(\pi x).
\end{equation}
\begin{figure}[ h t]
	\begin{minipage}{0.75\textwidth}
	\centering
	\includegraphics[width=0.7\textwidth]{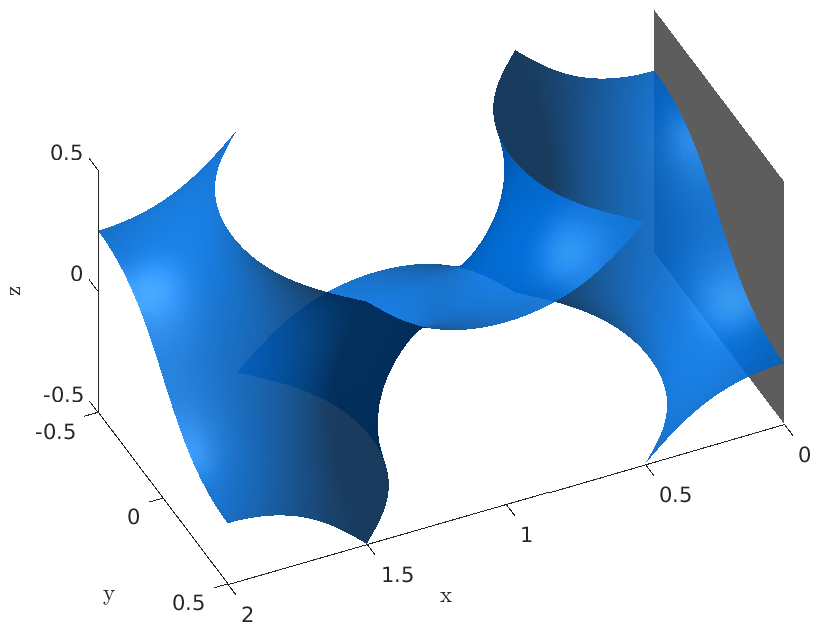}
	\end{minipage}
	\begin{minipage}[ht]{0.19\textwidth}
		\vspace{-2cm}
		\begin{align*}
		E  &=\Nmm{7 \cdot 10^{10}} \\
		\nu&=0.3  \\
		t &=\um{0.03} \\
		\end{align*}
	\end{minipage}
	\caption{Geometry of the gyroid problem. The structure is clamped at the gray plane.}
	\label{fig::ex_gyroid}
\end{figure} 
The considered shell lies in the cuboid $\cuboid = [\um{0},\um{2}]\times$ $[\um{-0.5},\um{0.5}]\times$ $[\um{-0.5},\um{0.5}]$. The geometry and the material parameters are depicted in \Cref{fig::ex_gyroid}. The shell structure is clamped at the boundary curve which is in the plane $x=\um{0}$. We assume a thickness $t= \um{0.03}$. We study the vertical deflection due to a volume load $\bodyforce=-10^7 \mathbf e_z \, \Nm{}$ at the point $[2,0.5,-0.25]$. The deformed geometry is depicted in \Cref{fig::ex_gyroid_solution}. The results of the proposed methods are summarized in \Cref{tab::gyroid_dispA}. We observe that the results obtained by the null-space method and the penalty method are nearly the same and that they are in good agreement with the reference displacement $u_z = -1.8812$ given in \cite{gfrerer2018b}. We remark that the reference solution was obtained for a seven-parameter shell model including more physical effects and thus leading to a slightly larger displacement. Therefore, the deviation in the deflection is acceptable. However, the  results obtained by the Lagrange multiplier method are incorrect. This issue should be further investigated in future work. 
\begin{figure}[h t]
	\centering
	\includegraphics[width=0.79\textwidth]{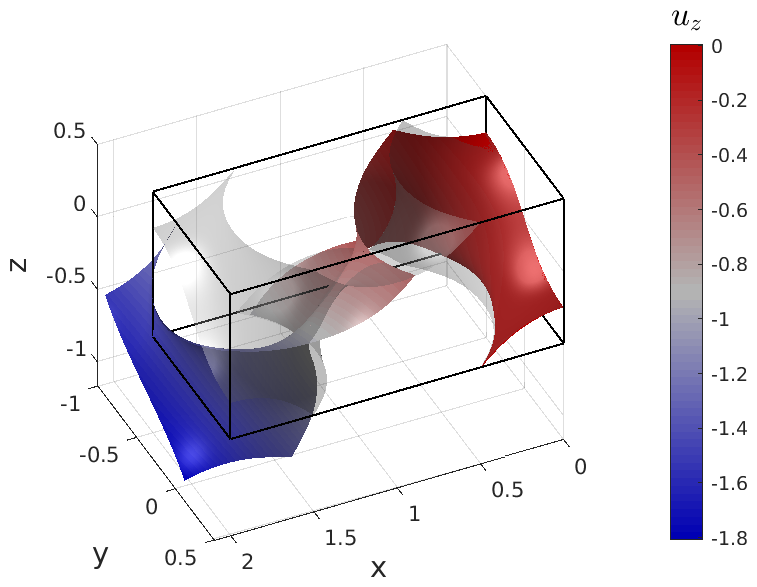}
	\caption{Deformed configuration of the gyroid}
	\label{fig::ex_gyroid_solution}
\end{figure}
\begin{table}
\centering
\caption{Displacements $u_z$ of the gyroid at the point [2,0.5,-0.25].}
\input{matlab_pics/gyroid_results.tex}
\label{tab::gyroid_dispA}
\end{table}

%% file: matlab_pics/roof_results.tex
\begin{tabular}{c|ccc}
ref. level & null-space & Lagrange mulitplier & penalty method \\ 
\hline 
0 & -0.43746 & -0.47767 & -0.29660 \\ 
1 & -0.30214 & -0.30239 & -0.30214 \\ 
2  & -0.30065 & -0.16734 & -0.30065 \\ 
3 & -0.30060 & -0.30059 & -0.30060 \\ 
4 & -0.30059 & -0.30063 & -0.30059 \\ 
5  & -0.30059 & -0.31847 & -0.30060 \\ 
\hline 
\end{tabular}

%% file: matlab_pics/hemisphere_results.tex
\begin{tabular}{c|ccc}
ref. level & null-space & Lagrange mulitplier & penalty method \\ 
\hline 
0 & 0.04490 & 0.04490 & 0.04490 \\ 
1 & 0.08989 & 0.08989 & 0.08990 \\ 
2  & 0.09222 & 0.09222 & 0.09222 \\ 
3 & 0.09239 & 0.09239 & 0.09239 \\ 
4 & 0.09241 & 0.09241 & 0.09240 \\ 
5 & 0.09241 & 0.09241 & 0.09241 \\ 
\hline 
\end{tabular}

%% file: matlab_pics/gyroid_results.tex
\begin{tabular}{c|ccc}
ref. level & null-space & Lagrange mulitplier & penalty method \\ 
\hline 
0 & -0.24147 & -0.54289 & -0.31639 \\ 
1 & -1.70309 & -2.05171 & -1.71521 \\ 
2  & -1.80865 & -2.40048 & -1.80900 \\ 
3 & -1.80891 & -2.61238 & -1.80925 \\ 
4 & -1.80905 & 3.77214 & -1.80318 \\ 
\hline 
\end{tabular}

%% file: s_appendix.tex
\begin{appendix}
\section{Proofs}\label{sec::proofs}
\begin{proof}[Proof of \Cref{thm::projectorMetric}] 
	It is sufficient to show that the application of the operators $\ParaMetric$ and $\projector\circ\parametrization=\mathbf I - \ParaNormalvector\otimes \ParaNormalvector$ to a basis $(\baseVector_l,\ParaNormalvector)$ gives the same result,
	\begin{align*}
	\mathbf G \cdot \baseVector_l &=  \baseVector_l, \\
	(\mathbf P\circ\parametrization) \cdot \baseVector_l &= \baseVector_l, \\
	\mathbf G \cdot \RealNormalvector &= \mathbf 0, \\
	(\mathbf P \circ\parametrization) \cdot \RealNormalvector &= \mathbf 0.
	\end{align*}
	Furthermore, direct calculation shows
	\begin{align*}
	\RealWeingarten \circ \parametrization = -(\nabla\RealNormalvector\cdot \projector) \circ \parametrization  = -(\nabla\RealNormalvector\circ \parametrization ) \cdot \baseVector_\indA \otimes \baseVector^\indA = -\ParaNormalvector_{,\indA}\otimes \baseVector^\indA = \ParaWeingarten.
	\end{align*}	
\end{proof}
\begin{proof}[Proof of \Cref{thm::surfaceDivergence}] 
	First we establish the relation
	\begin{equation}
	\sqrt{\det\ParaMetric}_{,\indC} = \sqrt{\det\ParaMetric} \;\christbS{\indC\indA}{\indA},
	\end{equation}
	Following \cite[Theorem 4.4-4]{ciarlet2006}, we have
	\begin{equation}
	\begin{aligned}
	\sqrt{\det(\ParaMetric)} = \det(\baseVector_1,\baseVector_2,\ParaNormalvector)
	\end{aligned}
	\end{equation}
	and the sought relation follows by
	\begin{equation}
	\begin{aligned}
	\left(\sqrt{\det(\ParaMetric)}\right)_{,\indA} &= \det(\baseVector_{1,\indA},\baseVector_2,\ParaNormalvector) + \det(\baseVector_1,\baseVector_{2,\indA},\ParaNormalvector) + \det(\baseVector_1,\baseVector_2,\ParaNormalvector_{,\indA}) \\
	&= \det(\christbS{1\indA}{\indD}\baseVector_{\indD}+\curvature_{\indA 1}\ParaNormalvector,\baseVector_2,\ParaNormalvector) +  \det(\baseVector_1,\christbS{2\indA}{\indD}\baseVector_{\indD}+\curvature_{\indA 2}\ParaNormalvector,\ParaNormalvector)  \\ &\qquad+ \det(\baseVector_1,\baseVector_2,-\curvature_{\indA}^\indD \baseVector_\indD) \\
	&= (\christbS{1\indA}{1} + \christbS{2\indA}{2}) \det(\baseVector_{1},\baseVector_2,\ParaNormalvector) \\
	&=\christbS{\indB\indA}{\indB}  \sqrt{\det(\ParaMetric)}. 
	\end{aligned}
	\end{equation}
	With \eqref{eq::derivativeBaseVector}, we obtain
	\begin{equation}
	\begin{aligned}
	 \left(\baseVector^\indA \sqrt{\det(\ParaMetric)}\right)_{,\indA} &= \baseVector^\indA_{,\indA} \sqrt{\det(\ParaMetric)} + \baseVector^\indA \left(\sqrt{\det(\ParaMetric)}\right)_{,\indA} \\ &= (-\Gamma_{\indA\indC}^\indA \baseVector^{\indC} + \curvature^{\indA}_{\indA} \ParaNormalvector)\sqrt{\det(\ParaMetric)} +\baseVector^\indA \christbS{\indB\indA}{\indB} \sqrt{\det(\ParaMetric)} \\
	&= H \ParaNormalvector\sqrt{\det(\ParaMetric)}.
	\end{aligned}
	\end{equation}
	Therefore, the first part of the lemma follows 
	\begin{equation}
	\begin{aligned}
	\text{div} \ParaTensor  &= \frac{1}{\sqrt{\det\ParaMetric}} \left(\ParaTensor\cdot\baseVector^\indA\sqrt{\det\ParaMetric}\right)_{,\indA} \\&= \ParaTensor_{,\indA}\cdot\baseVector^\indA + \frac{\left(\baseVector^\indA \sqrt{\det(\ParaMetric)}\right)_{,\indA}}{\sqrt{\det(\ParaMetric)}}\\
	&=\ParaTensor_{,\indA}\cdot\baseVector^\indA + H \, \ParaTensor \cdot \ParaNormalvector.
	\end{aligned}
	\end{equation}
	
	The second part of the lemma can be shown by direct calculation,
	\begin{equation}
	\begin{aligned}
	(\text{div} \RealTensor ) \circ \parametrization &= (\nabla \RealTensor \circ \parametrization) : (\baseVector_\indA\otimes\baseVector^\indA) + H \ParaTensor \cdot \ParaNormalvector \\ &= ((\nabla \RealTensor \circ \parametrization) \cdot \baseVector_\indA ) \cdot\baseVector^\indA + H \ParaTensor \cdot \ParaNormalvector	\\
	&=  \ParaTensor_{,\indA}\cdot\baseVector^\indA + H \, \mathbf T \cdot \RealNormalvector = \text{div} \ParaTensor.
	\end{aligned}
	\end{equation}
\end{proof}
\begin{proof}[Proof of \Cref{thm::productRules}] 
	The lemma can be shown by the direct calculations
	\begin{equation}
	\begin{aligned}
	\text{div}(\mathbf v\times \mathbf T)  &= \nabla_\Surface(\mathbf v \times \RealTensor):\projector +H \mathbf v \times \RealTensor \cdot \RealNormalvector \\
	&= 	[(\mathbf v \times \RealTensor)_{,i}\otimes \mathbf e^i\cdot \projector]:\projector +H \mathbf v\times\RealTensor  \cdot\RealNormalvector \\
	&= 	[(\mathbf v_{,i} \times \RealTensor + \mathbf v \times \RealTensor_{,i})\otimes \mathbf e^i\cdot \projector]:\projector +H \mathbf v\times  \RealTensor \cdot \RealNormalvector \\
	&= 	(-\RealTensor^\top\times\mathbf v_{,i}\otimes \mathbf e^i\cdot \projector + \mathbf v \times\nabla_\Surface \RealTensor ):\projector + H \mathbf v\times  \RealTensor \cdot \RealNormalvector \\
	&= \nabla_\Surface \mathbf v \ctimes \mathbf T^\top + \mathbf v \times \text{div}\RealTensor,
	\end{aligned}
	\end{equation}
	and
	\begin{equation}
	\begin{aligned}
	\text{div}(\mathbf v\cdot \mathbf T)  &= \nabla_\Surface(\mathbf v \cdot \RealTensor):\projector +H \mathbf v\cdot  \RealTensor \cdot \RealNormalvector \\
	&= 	[(\mathbf v \cdot \RealTensor)_{,i}\otimes \mathbf e^i\cdot \projector]:\projector +H \mathbf v\cdot\RealTensor \cdot \RealNormalvector \\
	&= 	[(\mathbf v_{,i} \cdot \RealTensor + \mathbf v \cdot \RealTensor_{,i})\otimes \mathbf e^i\cdot \projector]:\projector +H \mathbf v\cdot  \RealTensor \cdot \RealNormalvector \\
	&= 	(\RealTensor^\top\cdot\mathbf v_{,i}\otimes \mathbf e^i\cdot \projector + \mathbf v \cdot\nabla_\Surface \RealTensor ):\projector + H \mathbf v\cdot  \RealTensor \cdot \RealNormalvector \\
	&= \nabla_\Surface \mathbf v : \mathbf T^\top + \mathbf v \cdot \text{div}\RealTensor.
	\end{aligned}
	\end{equation}	
\end{proof}
\begin{proof}[Proof of \Cref{thm::changeTensors}]
	For the proof we use the relations \eqref{eq::relationDerivative} and \eqref{eq::relationDerivativeSecond}. Direct calculation yields for the linearized change in metric tensor
	\begin{equation}
	\begin{aligned} \RealChangeMetric\circ \parametrization &=
	\left[\frac{1}{2} \projector \cdot (\nabla \ureal + (\nabla \ureal)^\top) \cdot\projector \right] \circ \parametrization \\ 
	&= \frac{1}{2} (\baseVector^\alpha \otimes \baseVector_\alpha)\cdot(\nabla \ureal\circ \parametrization + (\nabla \ureal)^\top\circ \parametrization)\cdot(\baseVector_\beta \otimes \baseVector^\beta) \\
	&=\frac{1}{2} (\ureal_{,\beta} \cdot \baseVector_\alpha + \ureal_{,\alpha} \cdot \baseVector_\beta) \baseVector^\alpha \otimes \baseVector^\beta = \ParaChangeMetric,
	\end{aligned}
	\end{equation}	
	and for the linearized change in curvature tensor
	\begin{equation}
	\begin{aligned}
	\RealChangeCurvature \circ \parametrization&= [\projector\cdot(\RealNormalvector\cdot\nabla\nabla \mathbf u)\cdot\projector  - ( \RealNormalvector\cdot\nabla \ureal \cdot \RealNormalvector) \RealWeingarten]\circ \parametrization \\
	&= (\baseVector^\alpha \otimes \baseVector_\alpha)\cdot(\ParaNormalvector\cdot\nabla\nabla \mathbf u\circ \parametrization)\cdot(\baseVector_\beta \otimes \baseVector^\beta) 
	+ \ParaNormalvector\cdot(\nabla \ureal\circ \parametrization) \cdot \ParaNormalvector h_{\indA\indB} (\baseVector^\alpha \otimes \baseVector^\beta)   \\
	&= \ParaNormalvector\cdot \mathbf u_{,\indA\indB} \,\baseVector^\alpha \otimes \baseVector^\beta - \ParaNormalvector\cdot(\nabla \ureal \circ \parametrization) \cdot (\Gamma_{\indA\indB}^\indC \, \baseVector_\indC 
	+ h_{\indA\indB} \,\ParaNormalvector)\baseVector^\alpha \otimes \baseVector^\beta \\ 
	&\quad + \ParaNormalvector\cdot(\nabla \ureal \circ \parametrization)\cdot \ParaNormalvector \,h_{\indA\indB} (\baseVector^\alpha \otimes \baseVector^\beta) \\
	&= \ParaNormalvector\cdot( \mathbf u_{,\indA\indB} - \Gamma_{\indA\indB}^\indC \mathbf u_{,\indC}) \baseVector^\alpha \otimes \baseVector^\beta = \ParaChangeCurvature.
	\end{aligned}
	\end{equation}
\end{proof}
\begin{proof}[Proof of \Cref{thm::equilibirumMoments}]
	Applying the surface divergence theorem yields
	\begin{equation}
	\begin{aligned}
	\int_{\Surface} \text{div}(\n\times\mathbf M)  + \text{div}(\mathbf x \times \pmb \sigma) +\mathbf x \times \mathbf b \ddx &= 0.
	\end{aligned}
	\end{equation}
	Using the divergence product rule \eqref{eq::divergenceProductCross}
	results in
	\begin{equation}
	\begin{aligned}
	\int_{\Surface} \n\times\text{div}(\mathbf M)+ \nabla_\Omega \n\ctimes\mathbf M^\top  + \mathbf x \times\text{div}\pmb \sigma +\nabla_\Omega \mathbf x\ctimes \pmb \sigma^\top + \mathbf x \times \mathbf b \ddx &= 0.
	\end{aligned}
	\end{equation}
	With \eqref{eq::local_force_equilibrium},  and $\nabla_\Omega \mathbf x = \projector$ we have
	\begin{equation}
	\begin{aligned}
	\int_{\Surface} \n\times\text{div}(\mathbf M)+ \nabla_\Omega \n\ctimes\mathbf M^\top +\projector\ctimes \pmb \sigma^\top  \ddx &= 0.
	\end{aligned}
	\end{equation}
	Due to the definition of the stress tensor \eqref{eq::stressTensor} we have $
	\pmb \sigma^T = \mathbf N^T +  \mathbf S \otimes\RealNormalvector
	$ and it follows
	\begin{equation}
	\begin{aligned}
	\projector\ctimes \pmb \sigma^\top &= (\RealBaseVector_{\alpha}\otimes\RealBaseVector^{\alpha}) \ctimes (\mathbf N^T +  \mathbf S \otimes\RealNormalvector)\\
	&= (\RealBaseVector_{\alpha}\otimes\RealBaseVector^{\alpha}) \ctimes (N^{\indB\indC} \RealBaseVector_{\indC}\otimes\RealBaseVector_\indB+  \mathbf S^\indC \RealBaseVector_{\indC}\otimes\RealNormalvector) \\
	&= N^{\indB\indA}(\RealBaseVector_{\indA}\times\RealBaseVector_{\indB}) +  \mathbf S \times\RealNormalvector \\
	&= [\mathbf N^\top]_{\times} +  \mathbf S \times\RealNormalvector,
	\end{aligned}
	\end{equation}
	and furthermore
	\begin{equation}
	\begin{aligned}
	\nabla_\Omega \n\ctimes\mathbf M^\top &= (\RealNormalvector_{,\indA}\otimes \RealBaseVector^\indA)\ctimes(M^{\indB\indC}\RealBaseVector_\indB\otimes \mathbf \RealBaseVector_\indC) \\
	&= -(h_{\indA}^\indD \RealBaseVector_\indD\otimes \RealBaseVector^\indA)\ctimes(M^{\indB\indC}\RealBaseVector_\indB\otimes  \RealBaseVector_\indC) \\
	&= -h_{\indB}^\indD M^{\indB\indC} \RealBaseVector_\indD \times \RealBaseVector_\indC\\
	&= [-\mathbf H\cdot\mathbf M]_{\times}.
	\end{aligned}
	\end{equation}
	Thus,
	\begin{equation}\label{eq::balance_dirctor_momentum}
	\begin{aligned}
	&\int_{\Surface} \n\times\text{div}(\mathbf M) + \nabla_\Omega \n\ctimes\mathbf M^\top +\projector\ctimes \pmb \sigma^\top  \ddx   \\ &\quad=\int_{\Surface} \n\times(\text{div}(\mathbf M)-\mathbf S) + [-\mathbf H\cdot\mathbf M + \mathbf N^\top]_{\times} \ddx = 0.
	\end{aligned}
	\end{equation}
	From \eqref{eq::balance_dirctor_momentum} we deduce the sought conditions
	\begin{equation}
	[-\mathbf H\cdot\mathbf M + \mathbf N^\top]_{\times}  = 0,
	\end{equation}
	and
	\begin{equation}
	\mathbf S = \projector \cdot \text{div}(\mathbf M).
	\end{equation}	
\end{proof}

\stuff{	
\begin{proof}[Proof of \Cref{the::adjoint}]Direct calculation shows
		\begin{align}
		\mathbf T :(\mathbf M\cdot\mathbf H) &= (T_{\indA\indB} \baseVector^{\indA}\otimes\baseVector^{\indB}) : (M^{\indC\indD} h_\indD^\indE \baseVector_\indC\otimes\baseVector_{\indE}) = T_{\indA\indB} M^{\indB\indD} h_\indD^\indA,
		\end{align}
		and
		\begin{align}
		(\mathbf H\cdot\mathbf T) : \mathbf M &= (h_\indA^\indB T_{\indB\indC} \baseVector^{\indA}\otimes\baseVector^{\indC}) : ( M^{\indD\indE} \baseVector_\indD\otimes\baseVector_{\indE}) = h_\indA^\indB T_{\indB\indD}  M^{\indD\indA}.
		\end{align}
		Relabelling of indices yields the equivalence
		\begin{align}
		(\mathbf H\cdot\mathbf T) : \mathbf M &= T_{\indB\indD}  M^{\indD\indA} h_\indA^\indB = T_{\indE\indD}  M^{\indD\indF} h_\indF^\indE = T_{\indA\indB}  M^{\indB\indD} h_\indD^\indA.
		\end{align}
\end{proof}
\begin{equation}
\mathbf H\cdot\mathbf M = h_\indA^\indB \baseVector_{\indB} \otimes \baseVector^\indA \cdot M^{\indD\indC} \baseVector_\indD\otimes\baseVector_{\indC} = h_\indA^\indB M^{\indA\indC} \baseVector_\indB\otimes\baseVector_{\indC}
\end{equation}
\begin{equation}
\mathbf M\cdot\mathbf H = M^{\indD\indC} \baseVector_\indD\otimes\baseVector_{\indC} \cdot h_\indA^\indB \baseVector^\indA \otimes  \baseVector_{\indB}= h_\indA^\indB M^{\indC\indA} \baseVector_\indC\otimes\baseVector_{\indB}
\end{equation}
}
\section{Derivation of the weak form}\label{sec::derivation_weak}
In this section the weak form of the governing equations is derived. To this end, we multiply \eqref{eq::local_force_equilibrium} with a test function $\mathbf v \in \mathcal V_0$ and integrate over the shell surface,
\begin{equation}\label{eq::weak}
\begin{aligned}
-\int_\Surface  \mathbf v \cdot \text{div} \, \pmb \sigma \ddx &= \int_\Surface  \mathbf v \cdot \bodyforce \ddx.
\end{aligned}
\end{equation}
Here, the function space of the test functions is 
\begin{equation}
\begin{aligned}
\mathcal V_0 = \{&\pmb \eta: \Omega\rightarrow \R^3 \,|\,\pmb \gamma(\pmb \eta) \in L^2(\Surface,\R^3),\, \pmb \rho(\pmb\eta) \in L^2(\Surface,\R^3),\,
\\ &\qquad\qquad\pmb \eta \cdot \mathbf e_i = 0 \text{ on } \Gamma_{D_i},\, \nabla_\Surface(\RealNormalvector\cdot\pmb \eta) \cdot \pmb \mu = 0 \text{ on } \Gamma_{D_c} \},
\end{aligned}
\end{equation}
where $\Gamma_{D_i}$, $\Gamma_{D_t}$, and $\Gamma_{D_\mu}$ denote Dirichlet boundaries. On $\Gamma_{D_i}$ the displacement in direction $\mathbf e_i$ is restrained and on $\Gamma_{D_t}$ and $\Gamma_{D_\mu}$ the rotation of the shell around the boundary tangent vector and the boundary normal vector is restrained respectively. The corresponding Neumann boundaries are given by $\Gamma_{N_i} = \Gamma \setminus \Gamma_{D_i}$, $\Gamma_{N_t} = \Gamma \setminus \Gamma_{D_t}$, and $\Gamma_{N_\mu} = \Gamma \setminus \Gamma_{D_\mu}$.
Integration by parts of the term on the left side yields
\begin{equation}
\begin{aligned}
\int_\Surface  \nabla_\Surface\mathbf v \cdot \pmb \sigma^\top \ddx &= \int_{\Gamma} \mathbf v \cdot \pmb \sigma \cdot \pmb \mu \dd s_x + \int_\Surface  \mathbf v \cdot \bodyforce \ddx.
\end{aligned}
\end{equation}
We have $\pmb \sigma^\top = \mathbf N^\top + \mathbf S \otimes \RealNormalvector$ and obtain
\begin{equation}
\begin{aligned}
\int_\Surface  \nabla_\Surface\mathbf v \cdot \mathbf N^\top \ddx + \int_\Surface  (\RealNormalvector \cdot\nabla_\Surface\mathbf v) \cdot \mathbf S   \ddx &= \int_{\Gamma} \mathbf v \cdot \pmb \sigma \cdot \pmb \mu \dd s_x + \int_\Surface  \mathbf v \cdot \bodyforce \ddx.
\end{aligned}
\end{equation}
Using \eqref{eq::transverseForce} and integration by parts of the second term on the left yields
\begin{equation}
\begin{aligned}
\int_\Surface (\RealNormalvector\cdot\nabla_\Surface \mathbf v)  \cdot \mathbf S \ddx =
\int_\Gamma (\RealNormalvector\cdot\nabla_\Surface \mathbf v) \cdot \mathbf M \cdot \pmb \mu \dd s_x 
-\int_\Surface \nabla_\Surface(\RealNormalvector\cdot\nabla_\Surface \mathbf v) : \mathbf M  \ddx.
\end{aligned}
\end{equation}
Due to \eqref{eq::constitutionAgain} we obtain the relation 
\begin{equation}\label{eq::auxilI}
\begin{aligned}
\int_\Surface  \nabla_\Surface\mathbf v : \mathbf N^\top \ddx &= \int_\Surface  \nabla_\Surface\mathbf v : \bar {\mathbf N} \ddx - \int_\Surface  \nabla_\Surface\mathbf v : (\mathbf M \cdot \mathbf H)\ddx \\
&= \int_\Surface \pmb\gamma(\mathbf v):\mathcal E : \pmb\gamma(\mathbf u)  \ddx - \int_\Surface   (\mathbf H\cdot \nabla_\Surface\mathbf v) : \mathbf M \ddx.
\end{aligned}
\end{equation}
Furthermore, we have
\begin{align}\label{eq::auxilII}
-\left[ \mathbf H\cdot\nabla_\Surface\mathbf v + \nabla_\Surface(\RealNormalvector\cdot\nabla_\Surface \mathbf v)\right] : \mathbf M &= - \pmb\rho(\mathbf v): \mathbf M,
\end{align}
and
\begin{align}\label{eq::auxilIII}
\int_\Gamma (\RealNormalvector\cdot\nabla_\Surface \mathbf v) \cdot \mathbf M \cdot \pmb \mu \dd s_x - \int_{\Gamma} \mathbf v \cdot (\mathbf H \cdot \mathbf M)  \cdot \pmb \mu \dd s_x &= \int_\Gamma \nabla_\Surface (\RealNormalvector\cdot \mathbf v) \cdot \mathbf M \cdot \pmb \mu \dd s_x.
\end{align}
Therefore, by using \eqref{eq::auxilI}, \eqref{eq::auxilII}, and \eqref{eq::auxilIII} we obtain from \eqref{eq::weak} the final weak form 
\begin{equation}
\begin{aligned}
t\int_\Surface  & \pmb\gamma(\mathbf v):\mathcal E : \pmb\gamma(\mathbf u) \ddx + \frac{t^3}{12} \int_\Surface  \pmb\rho(\mathbf v):\mathcal E : \pmb\rho(\mathbf u)  \ddx = \int_\Surface \mathbf v \cdot \mathbf b \ddx\\ 
& + \int_{\Gamma_{N_i}} v_i \, N^N_i \dd s_x - \int_{\Gamma_{N_t}} \nabla_\Surface (\RealNormalvector\cdot\mathbf v) \cdot \mathbf t \; M^N_t \dd s_x - \int_{\Gamma_{N_\mu}} \nabla_\Surface (\RealNormalvector\cdot\mathbf v) \cdot \pmb \mu \; M^N_\mu \dd s_x.
\end{aligned}
\end{equation}

\end{appendix}